\documentclass[journal]{IEEEtran}
\usepackage{}

\usepackage{color}

\usepackage{cite}

\ifCLASSINFOpdf
   \usepackage[pdftex]{graphicx}
   \graphicspath{{../pdf/}{../jpeg/}}
   \DeclareGraphicsExtensions{.pdf,.jpeg,.png}
\else
   \usepackage[dvips]{graphicx}
   \graphicspath{{../eps/}}
   \DeclareGraphicsExtensions{.eps}
\fi

\usepackage{subfigure}

\usepackage[cmex10]{amsmath}
\allowdisplaybreaks
\usepackage{stfloats}
\usepackage{float}
\usepackage{url}
\usepackage{paralist}
\usepackage{multirow}
\usepackage{wrapfig}
\usepackage{makecell}
\usepackage{amsfonts, amsmath, amssymb, epsfig, graphicx, psfrag, subfigure, cite}
\usepackage{itrans}

\newtheorem{proposition}{Proposition}
\newtheorem{theorem}{Theorem}

\newtheorem{corollary}{Corollary}

\begin{document}
	\title{
	%Multi-object Classification via Crowdsourcing with a Reject Option in the Presence of Spammers
Prospect Theory Based Crowdsourcing for Classification in the Presence of Spammers}
	
	\author{Baocheng~Geng,~\IEEEmembership{}        Qunwei~Li,~\IEEEmembership{}       and~Pramod~K.~Varshney,~\IEEEmembership{Life~Fellow,~IEEE}
		\thanks{This work was supported
        by NSF under Grant ENG 1609916}
		\thanks{B. \ Geng, Q.\ Li and P.~K.\ Varshney are with the Department of Electrical Engineering and Computer Science, Syracuse University, Syracuse, NY 13244 USA
			(e-mail: bageng@syr.edu; qli33@syr.edu; varshney@syr.edu).}}

\maketitle
%	{\color{red} type-I spammer--skips all the tasks, type-II spammer--answers all the tasks}
\begin{abstract}
We consider the $M$-ary classification problem  via crowdsourcing, where crowd workers respond to simple binary questions and the answers are aggregated via decision fusion. The workers have a reject option to skip answering a question when they do not have the expertise, or when the confidence of answering that question correctly is low. We further consider that there are spammers in the crowd who respond to the questions with random guesses. Under the payment mechanism that encourages the reject option, we study the behavior of honest workers and spammers, whose objectives are to maximize their monetary rewards. To accurately characterize human behavioral aspects, we employ prospect theory to model the rationality of the crowd workers, whose perception of costs and probabilities are distorted based on some value and weight functions, respectively. Moreover, we  estimate the number of spammers and employ a weighted majority voting decision rule, where we assign an optimal weight for every worker to maximize the system performance. The probability of correct classification and asymptotic system performance are derived. We also provide simulation results to demonstrate the effectiveness of our approach.
\end{abstract}


\begin{IEEEkeywords}
Classification, crowdsourcing, spammers, human behavioral analysis,  distributed inference, information fusion, prospect theory 
\end{IEEEkeywords}

\section{Introduction}
\IEEEPARstart{C}{rowdsourcing} has attracted intense interest in recent years as a new paradigm for distributed inference. It harnesses the intelligence of the crowd,  by exploiting the inexpensive and online labor markets in an effective manner \cite{TapscottW2006,Howe2008Crowdsourcing,TapscottW2010,Yuen2011Survey,Bollier2011,Hossfeld2014best,li2017does}. Crowdsourcing enables a new framework to utilize distributed human wisdom to solve problems that machines cannot perform well, like handwriting recognition, anomaly detection, voice transcription, and image labelling \cite{paritosh2011computer,kamar2012combining,burrows2013paraphrase,7052378}. While conventional group collaboration and cooperation frameworks rely heavily on a collection of experts in related fields, the crowd in crowdsourcing usually consists of non-experts. Therefore, the responses obtained from the crowd have diverse quality levels, which makes decision fusion in the problem of classification via crowdsourcing quite challenging. 

Although crowdsourcing has been applied in many applications,
the quality of the aggregated result 
is relatively low \cite{IpeirotisPW2010,allahbakhsh2013quality,mo2013cross} due to the following reasons. First, the worker pool is anonymous in nature, which may result in an unskilled and unreliable crowd \cite{VarshneyVV2014}. Second, the assumption that the workers are sufficiently motivated, extrinsically or intrinsically, to take part seriously in the crowdsourcing task, is highly questionable \cite{Varshney2012e,hirth2013analyzing}. Third, for the non-expert crowd to successfully complete the crowdsourcing work, some tasks are specifically designed to be composed of easy but tedious microtasks \cite{VempatyVV2014}, which might cause boredom and result in low-quality work. Finally, noisy and unreliable responses to the tasks cannot be detected and tagged before aggregation so that appropriate weights could be assigned to responses \cite{yue2014weighted}. For this reason, simple majority voting is widely used as the aggregation rule and it takes all of the answers (including the noisy and low quality ones) into account with the same weight \cite{ruta2005classifier}.

Different methods have been developed to deal with the above problems \cite{VarshneyVV2014,hirth2013analyzing,VempatyVV2014,yue2014weighted}, \cite{KargerOS2011b,KargerOS2011a,6891807,QuinnB2011,zhang2012reputation}. 
    In \cite{KargerOS2011b}, the authors decompose a complex task into simple binary questions that are easy for the workers in the crowd to accomplish. It is expected that very little knowledge would be needed to complete the microtasks, and typically common sense or observation is good enough for such microtasks. The authors in \cite{VempatyVV2014} employ taxonomy and dichotomous keys in the design of the simple binary questions and the optimal question ordering problem in crowdsourcing is considered in \cite{geng2018decision}. These schemes that break hard questions into simple ones lower the chance for the workers to make mistakes in responding to each of the questions.  Different decision fusion rules are developed in order to deal with the unreliability of the crowd and increase  the classification accuracy \cite{yue2014weighted,6891807}. We have proposed the use of coding and decoding algorithms for reliable classification
  with unreliable crowd workers \cite{VarshneyVV2014,VempatyVV2014}. The comparison between  group control and majority voting techniques are presented in \cite{hirth2013analyzing}, which suggests that majority voting is more cost-efficient on simple binary tasks.

In the past literature, crowd workers could only submit a definitive yes/no answer in responding to a binary microtask/question.
However, research in psychology \cite{dhar1997consumer} indicates a frequent tendency to select the reject option  (no choice) when the choice set offers several attractive alternatives but none that can be easily justified as the best. In such cases, the workers may be unsure in answering some of the questions because of their lack of expertise. For instance, phonemes in some languages are very hard to distinguish, especially for foreigners \cite{varshney2016language}. To avoid requiring workers to respond to microtasks beyond their expertise resulting in random guesses, we considered the optimal design of the aggregation rule in crowdsourcing systems where the workers are not forced to make a binary choice when they are unsure of their response and can choose not to respond \cite{li2016multi}. As shown in \cite{jung2014predicting}, the quality of label prediction can be improved by adopting a decision rejection option to avoid results with low confidence. The reject option has also been considered in machine learning and signal processing literatures \cite{1054406, Bartlett:2008:CRO:1390681.1442792, doi:10.1287/deca.1080.0119, 5967817}. With a reject option, the payment mechanism is investigated in crowdsourcing systems where the workers can also report their confidence about the submitted answers \cite{shah2015double}.

Additionally, in crowdsourcing systems, there can be greedy crowd workers, also known as spammers, who aim to earn more monetary rewards by answering as many questions as possible. They often submit random guesses independent of the questions being asked. The presence of spammers degrades the system performance and has posed  a threat to many crowdsourcing applications \cite{wang2012serf}. Currently, there are two categories of anti-spammer techniques in crowdsourcing: a priori reputation system \cite{demartini2012zencrowd,raykar2012eliminating} and a posteri quality control \cite{liu2013scoring,downs2010your,difallah2013pick}. The first method aims to manage a pool of honest workers with high reputation, so as to ensure the reliability of their answers. However, since the crowd is usually large, anonymous and transient, it is impractical to keep track of the workers' answers and build up a trust relationship. In quality control schemes, several verifiable (golden standard) questions are inserted for the workers to answer and those who do not perform well on these questions are identified as spammers. However, due to the heterogeneous expertise levels of the crowd workers and the subjective criterion used to decide on a spammer, it is easy to mistakenly identify an honest worker to be a spammer, which discourages the worker from participating in answering the questions next time. According to the study in \cite{alonso2011crowdsourcing}, it is better to treat all the crowd workers to be honest, than risking to identify honest workers as spammers. 

We study the presence of spammers in the context of a scenario where crowd workers answer questions with a reject option. The payment mechanism proposed in \cite{shah2015double}, which encourages the workers to skip/reject answering a question below some confidence threshold, is employed.
In this scenario, both the honest workers and the spammers choose either to `answer a question'\footnote{In this case, honest workers submit their true answers and the spammers submit random guesses.} or `skip a question' to maximize their monetary rewards.
In \cite{shah2015double}, the authors assume that the workers are rational decision makers in the sense that they can perceive the expected payoff of taking each decision without biases. However,
since the crowd workers/spammers are humans, they are subject to cognitive biases in decision making and have disparate behavioral properties. The Nobel Prize winning prospect theory (PT) proposed by Kahneman and Tversky \cite{kahneman1979prospect} provides a psychologically accurate description of human cognitive biases that includes diminishing marginal utility, risk seeking/aversion and asymmetric valuation of gains and losses. In realistic crowdsourcing applications, these behavioral properties can not be ignored while developing effective decision making strategies of the honest crowd workers and the spammers.

We employ PT to model the rationality of the crowd workers, and study their behavior while answering or skipping the microtasks/questions\footnote{The terms `microtask' and `question' are used interchangeably in the paper.} in crowdsourcing that has a reject option. Based on the behavioral difference between the honest workers and the spammers,  we design the optimal aggregation rule at the fusion center (FC) to combat the effects of spammers. The contributions of our work are two fold: 
\begin{itemize}
\item By applying PT to model human cognitive biases, we study the optimal behavior of the honest workers and the spammers based on the payment mechanism proposed in \cite{shah2015double}. This payment mechanism has been proved to be the only mechanism that satisfies the ``no-free-lunch'' rule and exhibits incentive compatibility.  We  find  that  the  spammers should  either  complete  or  skip  all  the  microtasks  in  order  to get the maximal reward. The statistical behavioral properties of the crowd determine whether the spammers should complete or skip all the microtasks. 

\item We provide methods for estimating the number of spammers that is used for weight assignment. We also design an optimal aggregation rule where the workers are assigned appropriate weights\footnote{A brief discussion on the weight assignment was presented in \cite{li2018optimal}. In this work, we  elaborate on details and explanations related to this model.}. Probability of correct classification and asymptotic performance of our method are derived.
\end{itemize}

It should be noted that our approach only requires the workers to respond to several microtasks in one session without identifying themselves. Hence, our proposed method can be employed in many applications where the workers remain anonymous and, therefore, it is not required to keep track of the workers' profiles. On the other hand, instead of detecting the spammers, we estimate and employ the number of spammers while designing the optimal counter-measure to ameliorate their effects. The spammers still get paid according to the payment mechanism. As a result,  the risk of declaring honest workers as spammers and preventing them from further participation is avoided.

The rest of the paper is organized as follows. In Section II, we introduce the problem of multi-object classification via crowdsourcing with a reject option. In Section III, by employing PT, we study the optimal behavior of the honest crowd workers when their objective is to maximize their monetary rewards.  In Section IV, we consider that there are spammers in the crowd and develop a weight assignment scheme for optimal decision aggregation at the FC. Asymptotic performance analysis of our proposed method is presented in Section V. We  conduct simulation results in Section VI and conclude our work in Section VII.

\section{Classification via crowdsourcing with a Reject Option}

We formulate the  classification problem via crowdsourcing with a reject option in this section. Assume that we have $W$ workers participating in an $M$-ary classification task. There are $N = \left\lceil {{{\log }_2}M} \right\rceil$ simple binary questions to be answered by each worker, where we consider that the binary questions are independent of each other and are of the same difficulty. For each of the questions, the worker can either provide a definitive answer ``1'' (Yes) / ``0'' (No) \cite{VempatyVV2014,rocker2007paper}, or has a reject option to skip the question, where a skipped answer is denoted as  $\lambda$. Let ${\bf a}_w$ represent the $N$-bit word that contains the $w^{\mbox{th}}$ worker's ordered answers to all the microtasks, where ${\bf a}_w(i)\in \{1,0,\lambda\}$ for $i=1,\dots,N$. We assume the following statistical properties for the honest workers in the crowd: let $p_{w,i}$ be the probability that the $w^{\mbox{th}}$ worker submits $\lambda$ to the $i^{\mbox{th}}$ question, i.e, ${\bf a}_w(i)=\lambda$, and let $r_{w,i}$ be the probability that ${\bf a}_w(i)$ is the correct answer to the $i^{\mbox{th}}$ question, given that the worker has provided definitive answers ``1'' or ``0''. Since the workers in the crowd are anonymous and have diverse expertise levels, we consider that $p_{w,i}$ and $r_{w,i}$ are random and follow certain probability density functions (PDFs) $f_p(p)$ and $f_r(r)$, respectively. 
%Since all the $N$ questions are of the same difficulty, we drop the sub-index $i$ and simply use $p_w, r_w$ instead of $p_{w,i},r_{w,i}$. 
The expected values of $p_{w,i}$ and $r_{w,i}$, namely, the average probability that a worker submits $\lambda$ to a question, and the average probability that an answer is correct given a definitive answer has been submitted, are denoted by $m$ and $\mu$, respectively.

%For ease of exploitation, we assume that the microtasks are equally likely to have true answer to be ``0'' and true answer to be ``0''.

%Let $H_0$ and $H_1$ denote the hypotheses where ``0'' or ``1'' is the true answer for a single microtask, respectively. For simplicity of performance analysis, $H_0$ and $H_1$ are assumed equiprobable for every microtask. The crowdsourcing task manager or a fusion center (FC) collects the $N$-bit words from $W$ workers and performs fusion based on an aggregation rule. 
%Let $\rho^{s,t}_{w,i}, s,t\in \{0,1\}$ denote the probability that ${\bf a}_w(i)=t$ under hypothesis $H_s$ given that $w$th worker has a definitive answer for the corresponding task.  and we assume $\rho^{s,t}_{w,i}=\rho^{1-s,t}_{w,i}$. Therefore, the correct probability $\rho_{w,i} = {\rho ^{1,1}_{w,i}} = {\rho ^{0,0}_{w,i}}$.

After the $N$-bit words regarding an object from all the workers are collected by the FC, the object needs to be classified to a class $d_j\in D$, $j=1,\dots,M$, where $D$ is the set of all the object classes and $d_j$ is the $j$th class. From the $N$-bit word ${\bf a}_w(i)$ submitted by the $w^{\mbox{th}}$ worker, one can infer the classification decision of the $w^{\mbox{th}}$ worker namely the subset of classes $D_w$\footnote{If all the responses from the $w^{\mbox{th}}$ worker are definitive, $D_w$ is a singleton. Otherwise, $D_w$ contains multiple classes.} to which the object belongs to. At the FC, each class $d_j$ inside $D_w$ increments its candidate score by the weight assigned to the $w^{\mbox{th}}$ worker $W_w$. After incorporating the responses from all the $W$ workers, the FC determines the class with the highest overall candidate score to be the final classification result:
\begin{align}\label{first}
d^*= \arg \mathop {\max }\limits_{{{d_j} \in {D}} }\left\{ \! \sum\limits_{w= 1}^W {{W_w}I_{{D_w}}\left\langle {{d_j}} \right\rangle } \!\right\}, j=1,\dots ,M,
\end{align}where $I_{{D_w}}\left\langle {{d_j}} \right\rangle=1$  if $d_j\in D_w$ and $I_{{D_w}}\left\langle {{d_j}} \right\rangle=0$ otherwise. The objective is to find the appropriate weight assignment $W_w$ for every worker in the crowd, so that the best classification performance can be achieved. One approach is to split the $M$-ary classification task into $N$ binary hypothesis testing problems, each of which determines a bit in the $N$-bit word. For each hypothesis testing problem, the Chair-Varshney rule gives the optimal weight as $W_w={\log \frac{{{r _{w,i}}}}{{1 - {r _{w,i}}}}}$ \cite{ChairV1986}. However, this requires the prior knowledge regarding $r_{w,i}$ for every worker, which is not available in practice. One may also  look into the minimization of the misclassification probability, for which a closed-form expression for $W_w$ cannot be derived due to the lack of prior knowledge of $p_{w,i}$ and $r_{w,i}$.

We developed a weight assignment scheme to optimize the crowd workers' weights \cite{li2016multi}:
\begin{equation}\label{problem}
	\begin{array}{l}
	\text{maximize}\ \ {E_C}\left[ {{\mathbb{W}}} \right]\\
	\text{subject to}\ \ {E_O}\left[ {{\mathbb{W}}} \right] = {K}
	\end{array}
	\end{equation}where ${E_C}\left[ {{\mathbb{W}}} \right]$ denotes the crowd's average weight contribution to the
correct class and  ${E_O}\left[ {{\mathbb{W}}} \right]$ denotes the average weight contribution to
all possible classes. We set $K$ to be a constant so that the portion of weight contribution to the correct
class is maximized while the weight contribution to all the classes remains fixed. By assuming that there are no spammers in the crowd, we showed that the weight assigned to the $w^{\mbox{th}}$ worker is $W_w=\mu^{-n}$, where $n$ represents the number of definitive answers the worker submits in total. This method significantly outperforms the conventional majority voting approach.

In this following sections, we investigate the impact of spammers on system performance. Based on the payment mechanism that encourages the workers to skip the questions about which they are not sure, we characterize the behavior of both honest workers and spammers in realistic environments where they are subject to cognitive biases while decision making. With this information, we estimate the number of spammers in the crowd and design the weight assignment strategy for every worker to ameliorate the impacts of the spammers and maximize the system performance.

\section{Behavior of the Honest Crowdworkers}

In this section, we consider that there are no spammers in the crowd and  explore the workers' behavior in answering or skipping a microtask.  We adopt the payment mechanism proposed in \cite{shah2015double}, which encourages the use of the reject option when the confidence of answering a question is low. The mechanism was proved to be the only incentive-compatible mechanism that satisfies the ``no-free-lunch'' axiom (``no-free-lunch'' axiom requires that the payment is minimum possible if all the answers attempted by the worker in the gold standard questions are wrong). Under this payment mechanism, we discuss the optimal behavior for the honest workers when they are rational decision makers and want to maximize their monetary rewards. Next, considering that the honest workers are human decision makers that are subject to cognitive biases in practice, we employ prospect theory to model their behavioral property and analyze their decision making strategy in realistic environments. Compared to the case where the workers are assumed to behave rationally, it is shown that the behavioral factors captured in PT may cause the humans to act quite differently. 

\subsection{Payment Mechanism and Optimal Behavior of Rational Crowd Workers}
The payment to the worker is based on  the answers that the worker gives to the $G$ gold standard questions (which are not known to the crowdworkers in advance). The goal of the mechanism is to incentivize the worker to skip the questions for which its confidence is lower than a threshold $T$, where confidence about an answer is the
probability of this answer being correct. The value of $T\in[0,1]$ is chosen \textit{a priori} based on factors such as the targeted performance quality. A larger value of $T$ leads to a higher probability that a  question is skipped (or equivalently, a lower probability that this question is answered). When $T$ is large, the answer has a higher probability of being correct  given that a definitive answer has been submitted. Let $f$ denote the payment rule, which is proposed in \cite{shah2015double} and is written as 
\begin{align}\label{payment}
f(x_1,\ldots,x_G)=\kappa\prod\limits_{i=1}^G \alpha_{x_i}+\mu_{\min}
\end{align}
where $x_j\in \{-1,\lambda,+1\}, 1\le j \le G$, are the responses to the gold standard questions. ``$-1$'' denotes that the worker attempted to answer the microtask and the answer was incorrect, ``$\lambda$'' denotes that the worker skipped the microtask, and ``$+1$'' denotes that the worker attempted to answer the microtask and the answer was correct. Set $ \alpha_{-1}=0, \alpha_{\lambda}=1, \alpha_{+1}=\frac 1 T$, and $\kappa =\left( \mu_{\max} -\mu_{\min}\right) T^G$ with budget parameters $\mu_{\max}$ and $\mu_{\min}$ denoting the maximum and minimum payments respectively. Note that $\mu_{\min}$ is a constant that represents the fixed reward, and $\kappa \prod\limits_{i=1}^G \alpha_{x_i}$ represents the variable reward that is determined by the worker's answers. According to the payment mechanism, the variable reward is multiplied by a factor of $\frac 1 T$ when the worker answers the microtask and the answer is correct, and the variable reward reduces to $0$ if the answer is wrong. Therefore, only when the worker's confidence towards a microtask is higher than $T$, the expected payoff is positive and it is beneficial for the worker to answer the microtask. Otherwise, there is a loss in expected reward if the answer to the microtask is incorrect and in this case the worker should use the skip option. Note that when the  confidence is exactly  equal to $T$, the worker can choose to either answer or skip the microtask. Under this payment mechanism, the workers are encouraged to use the reject option when they are not sure of their answer, i.e., when the confidence regarding the question is smaller than $T$. Next, we employ PT to model the rationality of the crowd workers in deciding to answer or skip a microtask.

\subsection{Behavior of Crowdworkers Predicted by Prospect Theory}

In this subsection, we analyze the behavior of the honest workers while considering their cognitive biases. We start with a brief introduction of PT. According to PT, human cognitive biases have the following two characteristics: (i) a human decision maker is loss averse in that he/she  strongly prefers avoiding losses than achieving gains, (ii) humans are risk seeking in events with small probabilities and are risk averse in events with large probabilities, which can be interpreted as people overweighting small probabilities and underweighting large probabilities. 
Under PT, humans are loss averse and have asymmetric valuation towards gains and losses through a value function: 

\begin{equation}\label{eq:value function}
v(x)=\left\{
\begin{array}{rcl}
x^\delta       &      & {x \geq 0}\\
-\beta(-x)^\delta     &      & {x < 0}
\end{array} \right.
\end{equation}where $v(x)$ assigns a subjective utility to an outcome $x$. A positive $x$ represents a gain and a negative $x$ represents a loss.  As $\beta$ changes, $v(x)$ reflects different loss aversion attitudes of humans. Under PT, gains and losses are perceived with respect to a reference point, which is usually the current wealth and is set to be $0$ in (\ref{eq:value function}). $\delta$ is the coefficient of diminishing marginal utility and it suggests that as more gains are added (or more losses are subtracted) to the present wealth, the subjective utility of an additional gain (or loss) becomes more insignificant. The value function $v(x)$ is illustrated in Fig. 1(a).

\begin{figure}[htb]
  \centering
  \subfigure[Value function]{ 
    \includegraphics[width=0.23\textwidth]{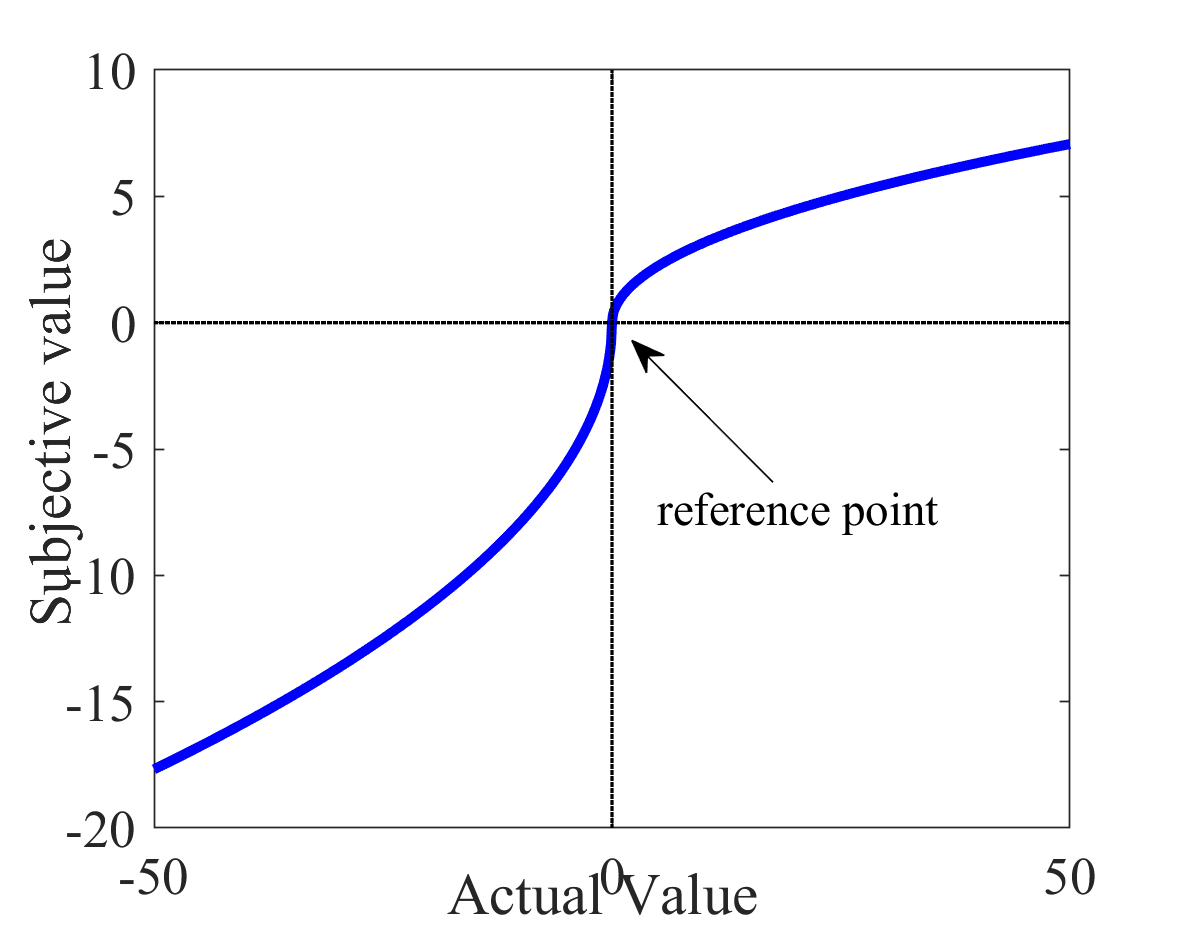}}
%  \hspace{1in}
  \subfigure[Weight function]{
    \includegraphics[width=0.23\textwidth]{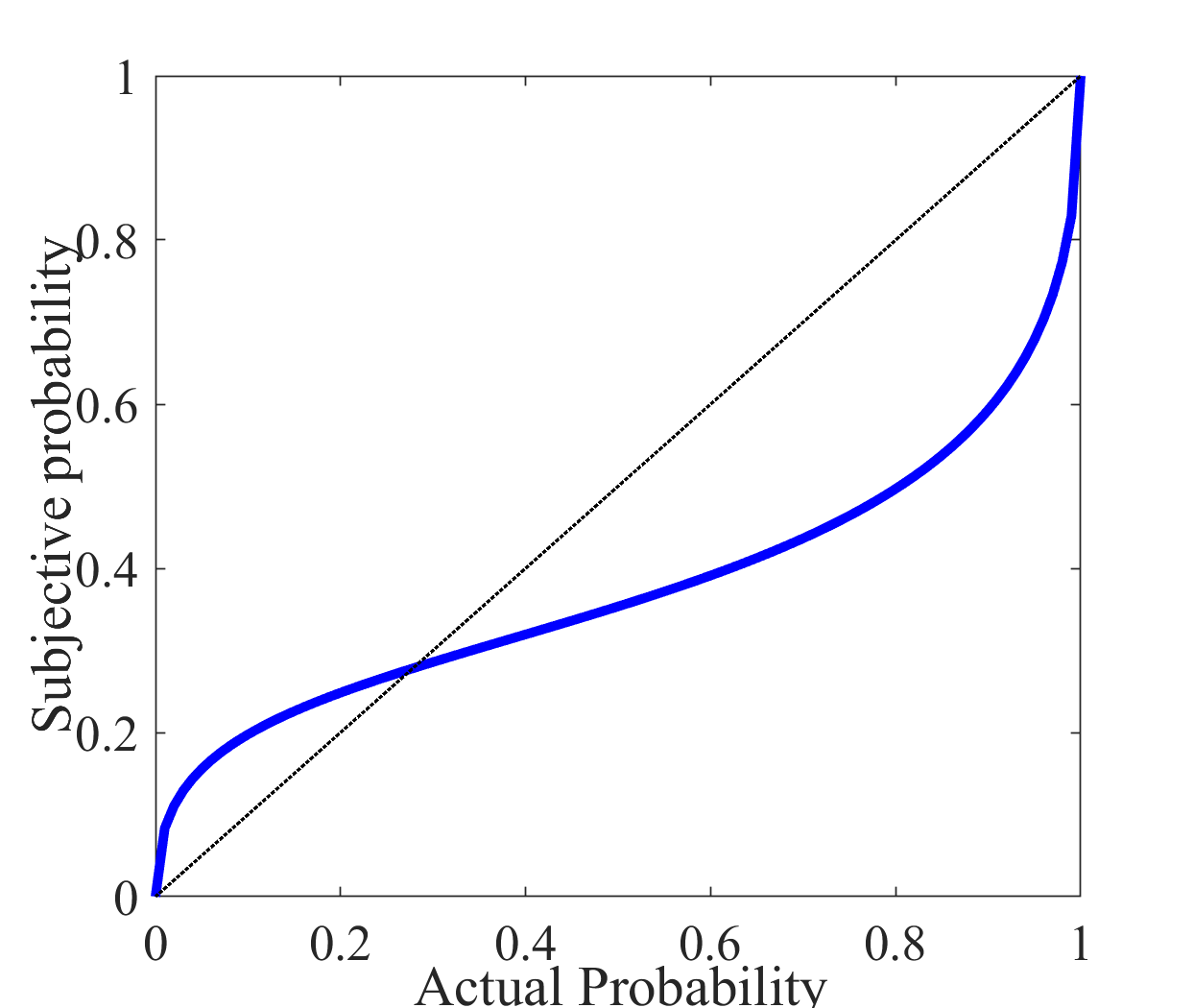}}
  \caption{Illustration of the value function and the weight function in prospect theory}
  \label{fig:subfig} %% label for entire figure
\end{figure}

As shown in Fig. 1(b), the probability weighting function shows humans' distorted perception of probability with which an event occurs:

\begin{equation}\label{eq:weight function}
w(x) = \frac{x^{\alpha}}{(x^{\alpha}+(1-x)^{\alpha})^{1/\alpha}}
\end{equation}where $x$ is the actual probability and $\alpha$ is the probability distortion coefficient. A smaller value of $\alpha$ means more severe distortion of the actual probability $x$. For rational decision makers, the PT behavioral parameters $\alpha=\beta=\delta = 1$. In practice, these parameters can be estimated through real experimental data and Tversky et al. showed that the means of $\alpha$, $\beta$ and $\delta$ are $0.69$, $2.25$ and $0.88$, respectively \cite{tversky1992advances}. 

According to the expected utility theory (EUT) \cite{mongin1997expected}, when a decision maker selects an action from a set of alternative choices, the one that results in a higher expected payoff is always preferred. In realistic environments, we employ PT to model the rationality of the crowdworkers and study their strategy in responding to the microtasks.  In our set up, the crowdworker has to decide whether to answer or skip a particular microtask based on the confidence $t$ , i.e., the probability of correctly answering the microtask. Under PT, a behaviorally biased crowdworker makes decisions according to the following theorem to maximize the subjective payoff.

\begin{theorem}
If the confidence of a
crowdworker with behavioral parameters $\alpha$, $\beta$ and $\delta$ towards a question is $t$, then he/she decides to answer or skip a question according to the following rule:
%\begin{equation}
%\frac{t}{1-t}\quad\mathop{\gtreqless}_{skip}^{answer}\quad\left(\frac{\beta T}{1-T}\right)^{\frac{\lambda}{\alpha}}\triangleq \eta.
%\end{equation}
\begin{align}\label{th1}
{\frac{t}{1-t}}\overset{{{answer}}}{\underset{{{skip}}}{\gtrless}}{{\left(\frac{\beta T}{1-T}\right)^{\frac{\delta}{\alpha}}\triangleq \eta. }}\end{align}
\end{theorem}
\begin{proof}
See Appendix \ref{prospect}
\end{proof}
Note that (\ref{th1}) can be written as:
\begin{align}\label{humanact}
{t}~\overset{{answer}}{\underset{{skip}}{\gtrless}}~{{t^*\triangleq\frac{\eta}{1+\eta}. }}
\end{align}For a rational decision maker with $\alpha=\beta=\delta=1$, $t^*=T$ and the decision rule (\ref{th1}) suggests answering the question if $t>T$, and skip the question otherwise. 
%\begin{equation}
%t\quad\mathop{\gtreqless}_{skip}^{answer}\quad t^*\triangleq\frac{\eta}{1+\eta}.
%\end{equation}
We find that when $\beta T \geq 1-T$, $t^*$ becomes larger, i.e., the worker is more likely to skip the question, as $\beta$, $\delta$ increase and $\alpha$ decreases. Otherwise, $t^*$ becomes smaller as $\delta$ increases and $\alpha,\beta$ decrease.

%For the crowd workers with average behavioral parameters $\alpha= 0.69$, $\delta= 0.88$ $\beta=2.25$ \cite{tversky1992advances}, $t^*=$

%Note that all the micro-questions are of the same difficult level, and
We consider that all the binary questions are equally difficult and the $w^{\mbox{th}}$ crowd worker has average confidence $t$ of answering a question correctly. Let  $t$ follow the PDF $f^w_t(t)$. According to the behavior characterized in (\ref{humanact}), the probability that the $w^{\mbox{th}}$ crowd worker skips a question can be expressed as $p_w=\int_0^{t^*}f^w_t(t)dt$. The probability that the answer is correct, given the $w^{\mbox{th}}$ worker has submitted a definitive  answer, can be expressed as $r_w=\int_{t^*}^1 f^w_t(t)dt$. Statistically,  $p_w$ determines the number of definitive answers $n$ of the $w^{\mbox{th}}$ worker and $r_w$ determines $\mu$. Recall that the weight assignment scheme in \cite{li2016multi} is $W_w(n) = \mu^{-n}$. The variations of $p_w$ and $r_w$ in the crowd may lead to different weights assigned to the workers. Given the confidence PDFs of all the workers $F_t = \{f_t^1(t),\dots,f_t^W(t)\}$ and  a priori $T$, the classification performance obtained assuming that the workers are rational is not accurate in predicting the system performance in realistic situations, where the crowd workers make decisions under cognitive biases. The first part of our simulations in Section VI provides comparisons between the system performance for a crowd with different behavioral properties. 
%However, if the parameters $m$ and $\mu$ are directly estimated from the workers' answers as in \cite{li2016multi}, behavioral properties of the crowd are inherently incorporated in the parameter estimation process and the prediction result will be accurate.

\section{Classification in the presence of spammers}
Spammers  are  known to exist in large numbers on crowdsourcing platforms. They submit their answers randomly without being relevant to the question being asked, in the hope of earning some extra money. In this section, we determine the optimal behavior of the spammers that maximizes their expected monetary reward based on the payment mechanism (3). By optimal behavior, we mean the optimal number of questions to be skipped by the spammers. 
We assume that a spammer skips $g$ out of $G$ gold standard questions, and answers the remaining $G-g$ by random guesses.

\subsection{Spammers are Rational}
First, we assume that the spammers behave rationally, i.e., they do not have cognitive biases while calculating the expected payoffs of answering or skipping a question. We hereby define the spammers that skip all the microtasks as \textbf{Type I} spammers, and the spammers that complete all the microtasks as \textbf{Type II} spammers. 
\begin{proposition}\label{op_spammer}
	To maximize the expected monetary reward,  a rational spammer completes all the microtasks (Type II) if $T<\frac 1 2$, and skips all the microtasks (Type I) otherwise.
\end{proposition}
\begin{proof}
See Appendix \ref{paymentA}.
\end{proof}

The above proposition gives the optimal strategy for the spammers to participate in the crowdsourcing task. Since a spammer can not distinguish the gold standard ones from the other questions, the result derived in Proposition 1 indicates that
the spammers should either complete or skip all the questions according to the value of $T$ to maximize their expected monetary reward. 
In realistic applications, the FC selects $T$ ($T>1/2$) to ensure the high quality of the workers' definitive answers (note that the workers' minimum possible value of the confidence regarding a question is $1/2$, when the answer is a random guess). In this case, according to Proposition 1, all the spammers are Type I and choose to skip all the microtasks.  It is shown later in the paper (Appendix C) that the answers from Type I spammers who respond $\lambda$ to all the microtasks are not aggregated for decision fusion. In other words, the weights assigned to Type I spammers do not affect the aggregation result and in this sense, we can consider all Type I spammers to be honest. As a result, by assuming that spammers behave rationally, all the workers in the crowd can be treated as honest ones and the weight assignment scheme is $W_w(n)=\mu^{-n}$.

\subsection{Spammers are Modeled by PT}
As the spammers are also human decision makers, we employ PT to model their rationality and predict their behavior in completing the micro-tasks. From the result of Theorem 1, we have the following corollary.

\begin{corollary}
To maximize the subjective monetary reward under PT, a spammer with behavioral parameters $\alpha,\beta,\delta$ completes all the microtasks (Type II) if $T>\frac{1}{\beta+1}$, and skips all the tasks (Type I) otherwise.
\end{corollary}
\begin{proof}
Since a spammer employs random guesses to respond to the microtasks, $t=1/2$. From the result of Theorem 1 and note $\alpha>0,\delta>0$, the corollary follows.
\end{proof}

Based on the above analysis of the spammers' behavior, we study the optimal weight assignment strategy at the FC and the classification performance in this section. As mentioned earlier, the workers/spammers in crowdsourcing systems have different backgrounds and are heterogeneous in their behavioral parameters $\alpha,\beta,\delta$. Considering that the spammers behave according to Corollary 1, we  need the loss aversion coefficient $\beta$ to predict whether the spammers answer or skip all the questions. Since the spammers remain anonymous in the crowd and certainly do not want to expose themselves, elicitation of parameter $\beta$ for the spammers is not possible. Therefore, without loss of generality, we assume that among the  crowd workers of size $W$, there are a total of $M=M_0+M_N$ spammers, with $M_0$  Type I spammers skipping all the microtasks and $M_N$  Type II spammers completing all the $N$ microtasks. 

%\subsection{Optimal Weight Assignment at the FC}

The presence of spammers will significantly affect the classification performance of the crowdsourcing system, which may make it even worse when the spammers are starting to act strategically. To ameliorate the spammers' impact on  system performance, we propose the aggregation rule for the FC by maximizing the candidate score assigned to the correct classification class as in (\ref{problem}). We denote our method as Amelioration of Spammers under PT (ASPT).
\begin{proposition}
In a crowd with $M_0$ Type I spammers and $M_N$ Type II spammers, the weight for the $w^{\mbox{th}}$ worker's answer  under formulation (\ref{problem}) is given by
\begin{align}\label{weight_assign}
{W_w(n)} = \left[\left( {W - M} \right){\mu ^n} +% \frac{{{M_0}}}{{{m^N}}}\delta \left( n \right)
  \frac{{{M_N}}}{{{2^N}{{\left( {1 - m} \right)}^N}}}\delta \left( {n - N} \right)\right]^{-1},
\end{align}
where $n$ represents the number of definitive answers submitted by the $w^{\mbox{th}}$ worker, and $\delta(\cdot)$ is the Dirac delta function.
\end{proposition}
\begin{proof}
See Appendix \ref{weightassignment2}.
\end{proof}

{Since the answers from Type I spammers who respond $\lambda$ to all the microtasks are not aggregated for decision fusion,  $M_0$ does not appear in the expression (\ref{weight_assign}).  In our scheme, the honest workers whose answers are not all definitive employ weights equal to $W_w(n)= \frac{1}{(W-M)\mu^n}$. This is the same as the weight assignment $W_w(n) = \mu^{-n}$ developed in \cite{li2016multi}, where all the workers are assumed to be honest\footnote{$W-M$ is a constant representing the number of honest workers in our scheme. Here, this constant acts as a scaling parameter.}. For workers who submit all definitive answers, the weight is decreased to a smaller value by adding $\frac{M_N}{2^N (1-m)^N}$ to the denominator of $\frac{1}{(W-M)\mu^n}$.  The weight assigned to a worker with all definitive answers, namely, $W_w(N)$ can not be large because it is likely that this worker is a spammer. On the other hand, since it is possible that this worker is honest, $W_w(N)$ can not be too small. Essentially maximization of the candidate score for the correct classification class gives the optimal value of $W_w(N)$, leading to the expression in (\ref{weight_assign}). The larger $M_N$ is, the more likely that a worker with all definitive answers is a spammer. Correspondingly, $W_w(N)$ is smaller.}

\section{Parameter estimation and performance analysis of ASPT}
In this section, we present the  parameter estimation technique used in our proposed method. Classification performance and asymptotic performance will also be examined.

\subsection{Parameter Estimation}

\begin{figure*}[!ht]
	\normalsize
	\begin{align}\label{mle}
f(W_{N+G},W_0|M_N,M_0      )=&\binom{W-M_0-M_N}{W_0-M_0}(\hat m^{N+G})^{W_0-M_0}(1-\hat m^{N+G})^{W-W_0-M_N}\nonumber\\
&\cdot \binom{W-W_0-M_N}{W_{N+G}-M_N}(1-\hat m)^{(N+G)(W_{N+G}-M_N)}\left(1-(1-\hat m)^{N+G}\right)^{W-W_{N+G}-W_0}
	\end{align}
		\hrulefill
\end{figure*}

The FC needs to estimate the crowd parameters $\mu$, $m$,  $M_N$, $M_0$ before assigning weights to the workers. Following our previous work \cite{li2016multi}, either the \textit{Training-based} or the \textit{Majority-voting based} method is adopted to estimate $\mu$. The estimate of $m$ is given by the ratio of the sum of skipped questions and all the questions attempted by the crowd. Since $m$ and $\mu$ represent statistical parameters for the honest workers in the crowd, the workers completing or skipping all the questions are not incorporated in the parameter estimation procedure to mitigate the impacts of spammers. The number of Type I and II spammers $M_0$ and $M_N$ are jointly estimated using the maximum likelihood estimation (MLE) method. $G$ gold standard questions are inserted into the $N$ classification questions, so that a worker responds to a total of $N+G$ questions. After answers from all the workers are collected by the FC,  we  count the number of workers submitting $N+G$ definitive answers and skipping all the microtasks, denoted be $W_{N+G}$ and $W_0$, respectively. Given the numbers of Type I and II spammers $M_0$ and $M_N$, the joint PDF of $W_{N+G}$ and $W_0$, $f(W_{N+G},W_0|M_N,M_0      )$,  is expressed in \eqref{mle}, where $\hat m$ is the estimated $m$, and $\binom{a}{b}=\frac{a!}{(a-b)!b!}$. 

Therefore, by the MLE method, the estimates of $M_0$ and $M_N$, which are denoted by $\hat M_0$ and $\hat M_N$ respectively, can be obtained as
\begin{align}\label{likelihood}
\left\{ {\hat  M_N,\hat M_0 } \right\}=\arg \mathop {\max }\limits_{\left\{{  M_N, M_0 }  \right\} \ge 0 }f(W_{N+G},W_0|M_N,M_0      ) .
\end{align}By writing $\mathbb{W}=W_0,\dots,W_i,\dots, W_{N+G}$ where $W_i$ is the number of workers submitting $i$ definitive answers for $i=0,\dots,N+G$, $M_N$ and $M_0$ can be more accurately estimated according to 
\begin{align}
\left\{ {\hat  M_N,\hat M_0 } \right\}=\arg \mathop {\max }\limits_{\left\{{  M_N, M_0 }  \right\} \ge 0 }f(\mathbb{W}|M_N,M_0      ) .
\end{align}However, the likelihood function $f(\mathbb{W}|M_N,M_0)$ becomes very complicated to compute. As we will see later in the simulation results section, the approach using (\ref{likelihood}) is sufficient to get a relatively accurate estimate of  $M_N$ and $M_0$. After the estimates of $\hat \mu$, $\hat{m}$, ${\hat M_N}$, and ${\hat M_0}$ are obtained, the FC can assign the appropriate weight to each worker based on (\ref{weight_assign}) and use the answers for aggregation.

\subsection{Performance Analysis}
In this section, we assume that there are $M_0$ Type I spammers and $M_N$ Type II spammers in a crowd population of size $W$. The spammers attempt to maximize their monetary rewards under the PT model as presented in Section IV. The probability of correct classification $P_c$ is investigated for the weight assignment scheme (\ref{weight_assign}). For simplicity, we assume that the prior probabilities  of the  true answers for each microtask to be “0” or “1” are equal. Note that we have a  correct overall classification only when all the $N$ microtasks are correctly labeled.

\begin{proposition}
	The probability of correct classification $P_c$ in the crowdsourcing system is 
	\begin{align}
	{P_c} =\Big[ \frac{1}{2} &+ \frac{1}{2}\sum\limits_{S} {\binom{W,M}{\mathbb{G}}} \left( {F(\mathbb{G}) - F^{\prime}(\mathbb{G})} \right) \nonumber\\&+\frac{1}{4}\sum\limits_{S^\prime} {\binom{W,M}{\mathbb{G}}} \left( {F(\mathbb{G}) - F^{\prime}(\mathbb{G})} \right)\Big] ^N
	\end{align}
	with
	\begin{align}
	&F({{\mathbb G}}) \nonumber=\\& m^{q_0}{\left(\frac{1}{2}\right)}^{M_N}\!\!\prod\limits_{n = 1}^N {{{\left( {1\! -\! \mu } \right)}^{{q_{\! - \!n}}}}{\mu ^{{q_n}}}{{\left(\! {\binom{N\!-\!1}{n\!-\!1}{{\left( {1 \!-\! m} \!\right)}^n}{m^{N - n}}} \right)}^{{q_{ \!-\! n}} \!+\! {q_n}}}} 
	\end{align}
	and
	\begin{align}
	&F^{\prime}({\mathbb G})=\nonumber\\&  m^{q_0}{\left(\frac{1}{2}\right)}^{M_N}\!\!\prod\limits_{n = 1}^N {{{\left( {1 \!-\! \mu } \right)}^{{q_n}}}{\mu ^{{q_{ \!-\! n}}}}{{\left(\! {\binom{N\!-\!1}{n\!-\!1}{{\left( {1 \!-\! m} \right)}^n}{m^{N - n}}}\! \right)}^{{q_{ \!-\! n}}\! +\! {q_n}}}} 
	\end{align}where
\begin{align}
	{{\mathbb G}} = \{ &({q_{ - N}},{q_{ - N + 1}}, \ldots {q_N},M_N^{\prime},M_N^{\prime \prime}):\nonumber\\
		&\sum\limits_{n =  - N}^N {{q_n} = W - M_N -M_0} ,M_N^{\prime}+M_N^{\prime \prime}=M_N \},
	\end{align} and $q_n$, $M_N^{\prime}$, and $M_N^{\prime \prime}$ take values from  natural numbers $\{0,1,\dots\}$,
	
	\begin{small}
	\begin{align}
	{S} \!=\! \left\{ {{{\mathbb G}}\!:\!\sum\limits_{n = 1}^N {({q_n} \!-\! {q_{ - n}})W_w(n)}  \!+\!(M_N^{\prime}\!-\!M_N^{\prime \prime})W_w(N) }>0 \right\},
	\end{align}
	\begin{align}
	{S}^\prime \! =\! \left\{ {{{\mathbb G}}\!:\!\sum\limits_{n = 1}^N {({q_n} \!-\! {q_{ - n}})W_w(n)}  \!+\!(M_N^{\prime}\!-\!M_N^{\prime \prime})W_w(N)  }=0 \right\},
	\end{align}
	\end{small}and $\binom{W,M}{\mathbb{G}} = \frac{{(W-M_0)!}}{M_N^{\prime}!M_N^{\prime \prime}!{\prod_{n =  - N}^N {{q_n}!} }}$.
\end{proposition}

\begin{IEEEproof}
	See Appendix \ref{pf2}.
\end{IEEEproof}

\subsection{Asymptotic Performance Analysis}

	\begin{figure*}[!ht]

\begin{align}\label{MandV}
\mathcal{M}= \frac{\left(2\mu-1\right)\left(1-m\right)}{\mu} \left(\frac{1-m}{\mu}+m\right)^{N-1}\nonumber+\frac{(W-M)\left(2\mu-1\right)\left(1-m\right)^NZ_M}{\left(W-M\right)\mu^N Z_M +M_N}
\\
\mathcal{V}=\frac{1-m}{\left(W-M\right)\mu^2}\left(\frac{1-m}{\mu^2}+m\right)^{N-1}+\frac{\left((W-M)(1-m)^N+M_N\right) Z_M^2}{\left((W-M)\mu^N Z_M+M_N \right)^2}-\frac{\mathcal{M}^2}{W-M}.
\end{align}
\hrulefill
\end{figure*}
In a practical situation, the number of workers for the crowdsourcing task is relatively large (normally in the hundreds). Then, it is of great value to investigate the asymptotic system performance when $W$ approaches infinity. Here, we give the asymptotic performance characterization for a large crowd, i.e., for a large $W$.

\begin{proposition}
As the number of workers $W$ approaches infinity, the probability of correct classification $P_c$ can be expressed as
	\begin{align}\label{pc}
	P_c=\left[Q\left( { - \frac{{{\mathcal{M}{{}}}}}{{\sqrt {{\mathcal{V}{{}}}} }}} \right)\right]^N,
	\end{align}
	where $Q(x) = \frac{1}{{\sqrt {2\pi } }}\int_x^{\infty}  {e^{\frac{{ - t^2 }}{2}} dt}$, and $\mathcal{M}$ and $\mathcal{V}$ are given in (\ref{MandV})
%\begin{align}
%M=\frac{\left(2\mu-1\right)\left(1-m\right)}{\mu} \left(\frac{1-m}{\mu}+m\right)^{N-1}-\frac{\left(2\mu-1\right)\left(1-m\right)^NZ_M}{\left(W-M\right)\mu^{2N}+Z_M\mu^N},
%\end{align}and
%\begin{align}
%V=\frac{1-m}{\left(W-M\right)\mu^2}\left(\frac{1-m}{\mu^2}+m\right)^{N-1}+\frac{M_N\left(W-M\right)\mu^{2N}-\left(1-m\right)^NZ_V}{\left(W-M\right)\mu^{2N}\left(\left(W-M\right)^2\mu^{2N}+Z_V\right)}-\frac{M^2}{W-M}.
%\end{align}
with $Z_M={2^N\left(1-m\right)^N}$.
\end{proposition}
\begin{IEEEproof}
See Appendix \ref{asym}.
\end{IEEEproof}

As stated above, the size of the crowd in practice can be fairly large and the asymptotic result derived in \eqref{pc} is a good characterization of the actual performance. We further consider that the percentages of Type I and Type II spammers in the crowd are $\gamma$ and $\epsilon$ respectively and give the analysis for the following two cases:
\begin{itemize}
    \item Case 1: $\lim\limits_{W\rightarrow \infty} \frac {M_O} W =\gamma >0$, $\lim\limits_{W\rightarrow \infty} \frac {M_N} W=\epsilon =0$. In this situation, we have $\mathcal{M}= \frac{\left(2\mu-1\right)\left(1-m\right)}{\mu} \left(\frac{1-m}{\mu}+m\right)^{N-1}\nonumber+\frac{\left(2\mu-1\right)\left(1-m\right)^N}{\mu^N }$, which is a constant given $m$, $\mu$ and $N$. $\mathcal{V}=\frac{1}{W(1-\gamma)}\left(\frac{1-m}{\mu^2}(\frac{1-m}{\mu^2}+m)^{N-1}+\frac{(1-m)^N}{\mu^N }-\mathcal{M}^2 \right)$. Note that $W(1-\gamma)$ represents the number of honest workers in the crowd. As there are more honest workers in the crowd, $\mathcal{V}$ becomes smaller and the probability of correct classification becomes larger. The expressions of $\mathcal{M}$ and $\mathcal{V}$ analytically show that as long as the number of honest workers is fixed, the number of Type I spammers have no impact on the system performance.
    
\item Case 2: $\lim\limits_{W\rightarrow \infty} \frac {M_O} W =\gamma=0$, $\lim\limits_{W\rightarrow \infty} \frac {M_N} W =\epsilon>0$. We have $\mathcal{M}= \frac{\left(2\mu-1\right)\left(1-m\right)}{\mu} \left(\frac{1-m}{\mu}+m\right)^{N-1}\nonumber+\frac{\left(2\mu-1\right)\left(1-m\right)^NZ_M}{\mu^N Z_M + \frac{\epsilon}{1-\epsilon} }$. As the percentage of Type II spammers in the crowd $\epsilon$ increases, $\mathcal{M}$ becomes smaller and the classification performance in terms of $P_c$ deteriorates. In this scenario, it is not easy to show the monotonicity of $\mathcal{V}$ with respect to $\epsilon$ and we rely on simulations to show that the probability of correct classification decreases as $\epsilon$ is larger.
\end{itemize}

\section{Simulation Results}

In the first part of this section, we consider that there are no spammers in the crowd.  Simulations are provided to illustrate how PT affect the workers' behavior and system performance. 

\subsection{Crowdsourcing Without Spammers}
We plot the actual confidence thresholds of cognitively biased crowdworkers $t^*$ with respect to the pre-designed threshold $T$ in Fig. 2 for humans with different behavioral properties. It can be observed that $t^*$ becomes larger as $T$ increases. Since we restrict that $T\geq 0.5$, $\beta T>(1-T)$ is satisfied. As a result, in Fig. 2(a) we see that $t^*$ becomes larger as $\beta$ increases and in Fig. 2(b) we see that $t^*$ becomes smaller when $\alpha$ increases. Note that in the green curve in the upper subplot, $\alpha= 0.69$, $\beta=2.25$ and $\delta= 0.88$ are the mean values of behavioral parameters of the humans from the experiment in \cite{tversky1992advances}. Hence, for this group of population, the green curve represents the average $t^*$ employed by the nominal cognitively biased workers. Since $t^*>T$, we can see that the workers are more likely to use the skip option in practice.

\begin{figure}[h]
	\centering
	\includegraphics[width=3.3in]{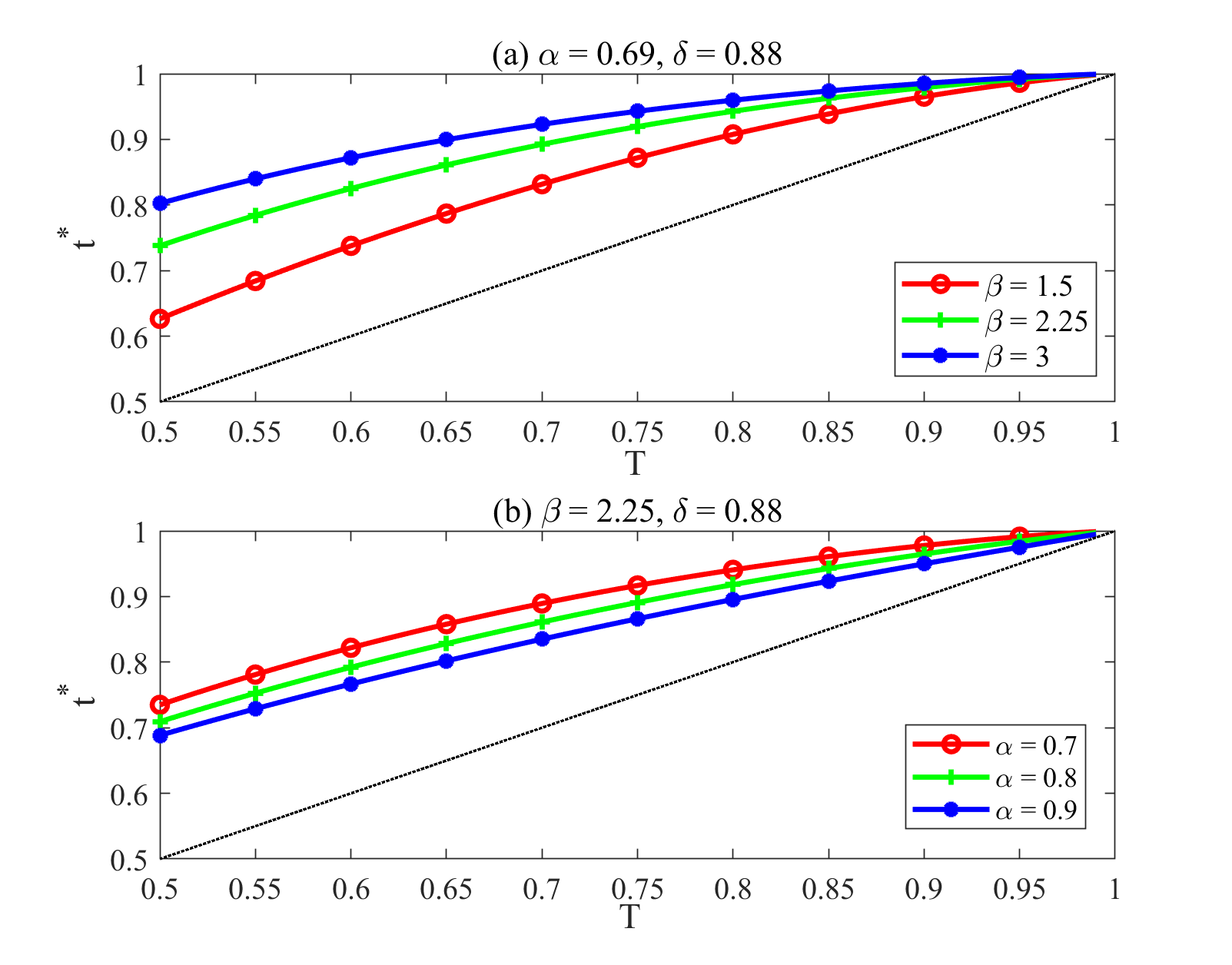} % requires the graphicx package
	\caption{Confidence thresholds of cognitively biased crowdworkers based on PT}
	\label{final1}
\end{figure}

Next, we assume that the confidence $t$ of the crowd workers follows a uniform distribution $U(0.5,x)$ where $x$ is uniformly distributed in $[0.7,0.9]$. The size of the crowd is $W=30$ and the confidence threshold $T$ is set equal to $0.6$. There are $N=3$ microtasks and $G=3$ gold standard questions. In Fig. \ref{final2} , we plot the system performance in terms of probability of correct classification $P_c$ for crowds with different behavioral parameters. If the workers in the crowd are assumed to be rational, i.e., $\alpha=\beta=\delta=1$, we obtain that $P_c = 0.8445$ and as the behavioral parameters change, $P_c$ has different values. Basically, given the distribution of the workers' confidence $t$, different behavioral parameters lead to different confidence threshold $t^*$, which in turn causes variations of the statistical parameters $m$ and $\mu$ of the crowd, leading to different classification performances. In the upper subplot where the probability distortion factor $\alpha=0.68$, $P_c$ first increases and then decreases as the loss aversion parameter $\beta$ becomes larger. Besides, $P_c$ is higher as the diminishing marginal utility parameter $\delta$ has a smaller value. In this case, appropriate values of $\beta$ and $\delta$ counter
the probability distortion factor $\alpha$ and improve the system performance. The behavioral parameters $\alpha, \beta$ and $\delta$ jointly determine the overall probability of correct classification $P_c$. Moreover, we can observe that the best achievable $P_c$ in this subplot is higher than $P_c=0.8445$ when the crowd workers are assumed to be rational. In the lower subplot where $\beta=2.25$, $P_c$ monotonically decreases as  $\alpha$ decreases from $1$ to $0.5$.  Same to the upper subplot, we have $P_c$ become larger as $\delta$ has a smaller value. 

\begin{figure}[ht]
	\centering
	\includegraphics[width=3.3in]{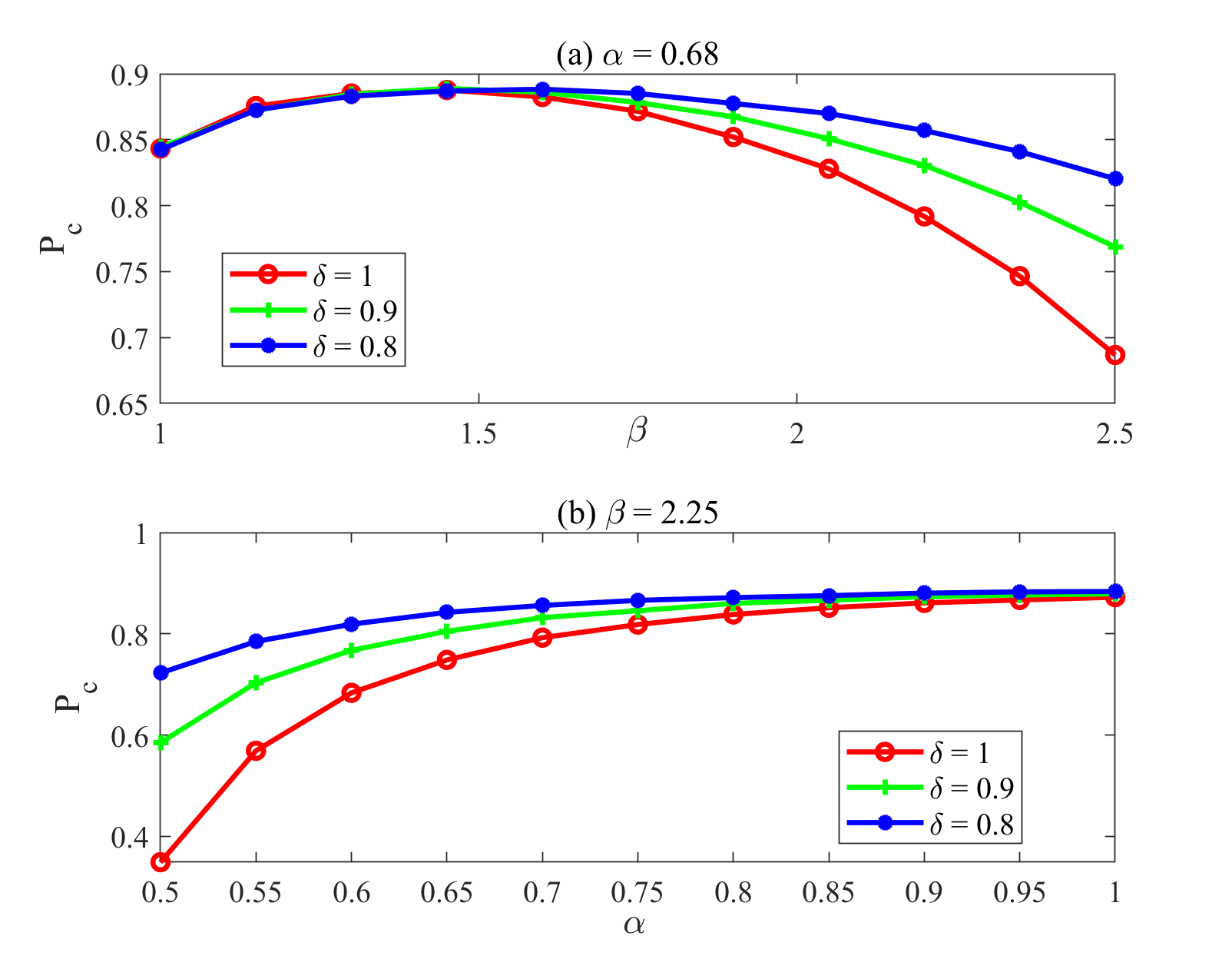} % requires the graphicx package
	\caption{Classification accuracy when crowd workers have different behavioral parameters.}
	\label{final2}
\end{figure}

\subsection{Crowdsourcing in the Presence of Spammers}
In this subsection, we present some simulation results to illustrate the advantage of our proposed method ASPT, in which  PT is employed to characterize the behavior of spammers. $W=50$ workers participate in a crowdsourcing task with $N=3$ microtasks and $G=3$ gold standard questions.  $f_p(p)$ is chosen as a uniform distribution $U(0.2,0.8)$, so that the average probability of a honest worker skipping a task is $m=0.6$. Let $f_r(r)$ be a uniform distribution expressed as $U(x,1)$ with $0\le x \le 1$, and thus we can have $\mu$ varying from 0.5 to 1.

\bgroup
\def\arraystretch{1.3}
\begin{table}[h]
	\centering
	\caption{Estimates of $M_0$ and $M_N$}
	\label{my-label}
	\setlength\tabcolsep{1.6pt}
\begin{tabular}{|c|c|c|c|c|c|c|c|c|c|c|}\hline
	\diaghead{\theadfont $M_0$~~~ $M_N$ 0.5}%
	{\ \ \ \ \ $M_N$}{$M_0$\ \ \ \ \ }
	&\thead{1}&\thead{3}&\thead{5} &\thead{7}&\thead{9}&\thead{11}&\thead{13}&\thead{15}&\thead{17}&\thead{19}\\   \hline
1	& 1,0  & 1,3  &  1,5   & 1,8   &  2,9  &  2,12  &  1,14  &  2,15  &  2,17  &2,20\\    \hline
3	&  3,1   &  3,2 &  3,5   &  4,7  &  4,9  &  3,11  &  3,14  &  3,15  &  3,18  &3,20\\    \hline
5	&  5,2   & 5,3  &  5,6  &   6,7  &  5,9  &  6,11  &  6,14  &  5,17  &  6,18  &  5,19\\  \hline
7	&   7,0  & 8,4  &  7,5  &   8,8  &  7,10  &  8,12  &  7,13  &  7,17  &  7,17  & 8,20\\    \hline
9	&  9,1   & 9,4  &  9,5   &  10,7  &  9,9  &  11,11  &  9,13  &  10,15  &  11,17  & 9,20\\    \hline
11	&  11,1   & 11,5  & 11,5   & 12,8   &  11,6  &  12,11   & 11,13   &  11,16  &  11,17  & 12,19\\    \hline
13	&   13,2  & 13,6 &  13,5   & 14,8   & 13,9   &  13,11   &  14,13  &  13,16  &  13,17  & 14,19\\    \hline
15	&  15,1   & 15,3  &  16,6  &  16,7  &  15,9  &  17,11   &  15,13  &  15,15  &  15,17  & 15,19\\    \hline
17	&  17,1   &  18,4 &  17,5   &  17,8  &  17,9  &  17,12  &  18,13  & 17,16  &  18,17  & 18,19\\    \hline
19	&  20,2   & 19,2  &  19,5  &  19,8   &  19,9  &  19,11  &  19,13  &  19,16  & 20,17   & 21,19\\    \hline
21	&   21,2  &  21,3 &  22,5   &  21,7  &  21,9  &  22,12  &  21,13  &  21,15  & 21,17   & 21,19\\    \hline
23	&  23,1   & 24,3  &   25,5  &  23,9  & 24,9   &  24,11  &  25,13  &  23,16  &  23,17  & 23,19\\    \hline
25	&  26,1   & 26,3 &  25,6    &  25,7  & 26,9   & 26,12   & 25,13   &  25,15  &  26,17  & 25,20\\    \hline
\end{tabular}
\end{table}
\egroup

First, we show the efficiency of our methods for estimating the parameters  $M_0$ and $M_N$. Table \ref{my-label} shows the estimation results of $M_0$ and $M_N$, when the true numbers of spammers are $M_0=\{1,3,\dots,19\}$ and $M_N=\{1,3,\dots,25\}$. Here, $\mu$ is set as 0.75. The estimation process is based on the distribution of the number of workers completing and skipping all the questions $W_{N+G}$ and $W_0$, and we can see from the table that most pairs of numbers $M_0$ and $M_N$ can be exactly estimated, and the estimation errors are at most $\pm 1$. 
\begin{figure}[ht]
	\centering
	\includegraphics[width=3.2in]{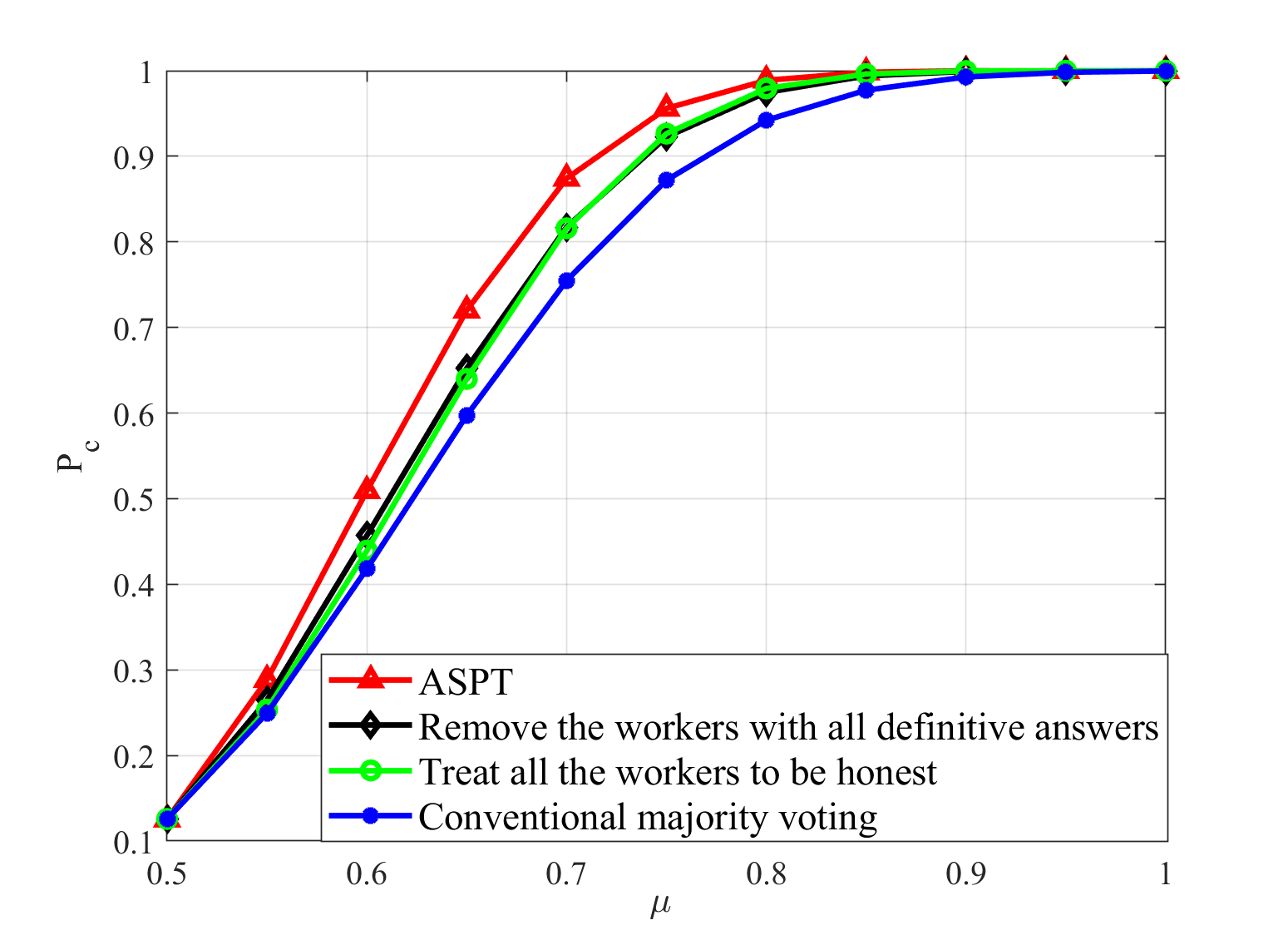} % requires the graphicx package
	\caption{Performance comparison with spammers as $\mu$ increases.}
	\label{Pcwithvariousmu_estimationofM0andMN}
\end{figure}

We present the performance comparison between different aggregation rules in Fig. \ref{Pcwithvariousmu_estimationofM0andMN}, where the quality of the crowd $\mu$ varies. For illustration, we assume that there are 14 spammers in a crowd of 50 workers, and we have 7 spammers completing all the questions and the other 7 skipping all the questions. As $\mu$ increases, we plot the probability of correct classification $P_c$ of four different weight assignment methods. The first one is the ASPT developed in this paper, where we employ PT for modeling the behavior of the spammers. % which is referred to as the optimal aggregation rule with spammers. 
In the second approach, we exclude the workers who submit all definitive answers and treat the remaining workers to be honest. The weights assigned to the honest workers are given by  $W_w=\mu^{-n}$ \cite{li2016multi}. The third one is where we consider the existence of spammer without incorporating PT, where all the spammers are assumed to be Type I and the FC treats all the workers to be honest, i.e., weight assignment rule is $W_w=\mu^{-n}$ no matter whether the workers submit all definitive answers or not. The last approach is conventional majority voting without a reject option, where all the workers are assigned the same weight. It can be seen in Fig.4 that at $\mu=0.5$, all the four curves merge to the same point. It is because when $\mu=0.5$, even the honest workers are making random guesses like a spammer. In this case, the FC collects no useful information from the crowd and the choice of weight assignment schemes does not make a difference. As the quality of the crowd, $\mu$, improves, the system performance also improves as expected. The proposed ASPT performs better than the method that excludes the workers with all definitive answers and the method that treats all the workers to be honest, which outperform the conventional majority voting approach that does not have a reject option. It should be noted that the second and the third methods have very similar performances. Compared to treating all the workers as honest ones, excluding the workers who submit all definitive answers has the advantage of removing the side effects of Type II spammers. At the same time, however, the second method may also remove the honest workers who submit all definitive answers for decision fusion, leading to a potential deterioration. Hence, this trade-off determines whether the second approach performs better than the third approach or not.

\begin{figure}[ht]
	\centering
	\includegraphics[width=3.2in]{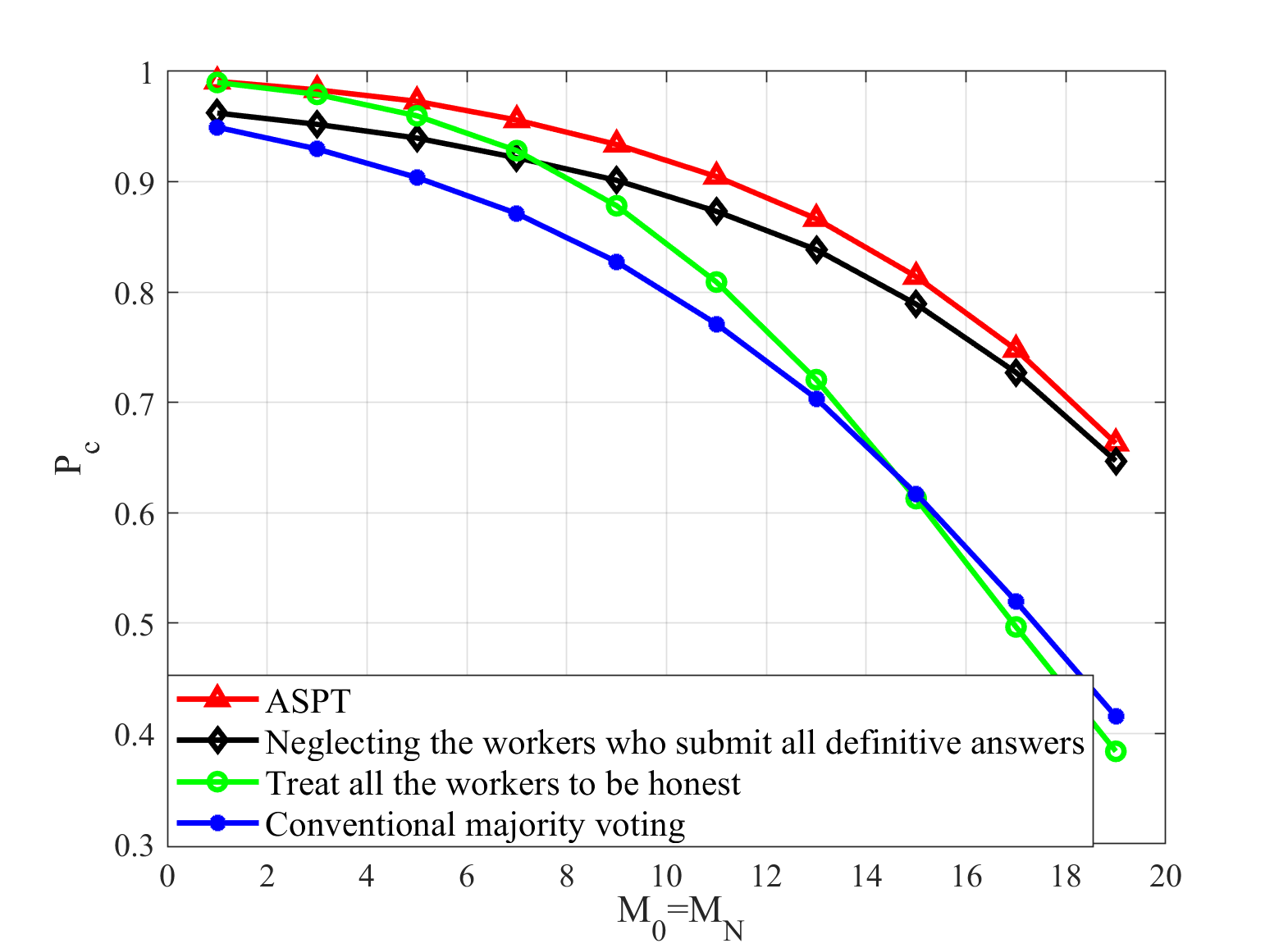} % requires the graphicx package
	\caption{Performance comparison with different numbers of spammers.}
	\label{Pcvariousspammers}
\end{figure}

In Fig. \ref{Pcvariousspammers}, we plot the performance comparison when the number of spammers changes. For simplicity, we set that $M_0=M_N$, and $\mu$ is fixed at 0.75. As the number of spammers increases, the classification performance degrades, where the ASPT method gives the best performance. Furthermore, there are two phenomena that need to be discussed:

1): When the number of spammers is small, the conventional majority voting method is outperformed by the one that treats all the workers as honest. However, this is not the case when the number of spammers is large. The reason is that with honest workers, the FC assigns a greater weight to the worker with a larger number of definitive answers. In the regime where $M_N$ is large, which means that the number of spammers completing all the questions is large, the impact from the spammers is much more severe on the performance with such a weight assignment scheme. Thus, the corresponding performance degrades significantly.

2): When the number of spammers is small, the method that excludes the workers who submit all definitive answers performs better than the method that treats all the workers to be honest, and vice versa. It can be explained by the fact that in the crowd when the percentage of spammers is small, the second method (the one that excludes the workers who submit all definitive answers) has a smaller probability to remove Type II spammers and has a higher probability to remove honest workers. On the other hand, when the percentage of spammers is large, the probability of excluding Type II spammers is large and that of excluding honest workers is small.

\begin{figure}[htb]
	\centering
	\includegraphics[width=3.2in]{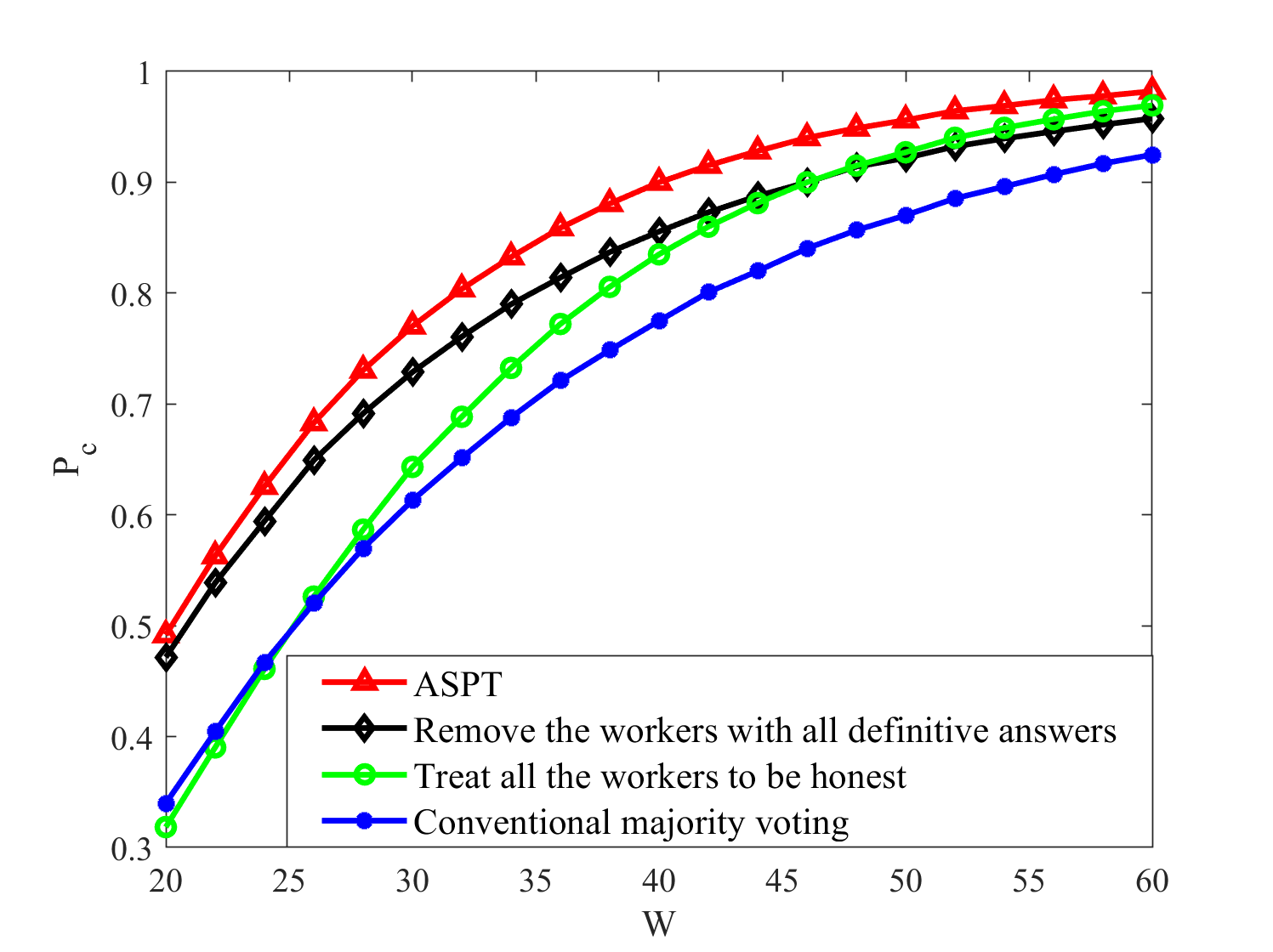} % requires the graphicx package
	\caption{Performance comparison as $W$ increases.}\label{lastfigure}
\end{figure}

Lastly, we keep the number of spammers in the crowd fixed, with $M_0=M_N=7$ and plot the asymptotic performance as the crowd population $W$ increases in Fig. \ref{lastfigure}. We can see that our proposed approach has the best performance among all the weight assignment strategies. Analogous to the explanations provided for results presented in Fig. \ref{Pcvariousspammers}, we observe that when $W$ is small, i.e., when the percentage of spammers is large, the method that excludes the workers with all definitive answers and the majority voting method outperform the method that treats all the workers to be honest, and vice versa.

\section{Conclusion}

We have explored a novel framework of crowdsourcing systems to solve classification problems, where the crowd workers may skip a microtask if the confidence regarding the question being asked is lower than a threshold. Our scheme is extremely effective in dealing with spammers as it: (i) exploits the behavioral differences of honest workers and spammers in realistic situations, where the rationality of humans is modeled via PT; (ii) estimates the number of spammers in the crowd and (iii) designs the optimal weight for every worker in the weighted majority voting fusion rule.
We provide analytical expressions for probability of correct classification and asymptotic system performance. Compared to rational decision makers,  the honest crowd workers and spammers behave in a different manner if PT is incorporated to model their rationality. To accurately characterize the system performance, the behavioral property of the crowd must be taken into consideration. Simulation results illustrate the  efficiency of our method compared to other weight assignment schemes that do not include the humans behavior using a prospect theoretic approach. 

This work employed a psychologically accurate description of human behavior in crowdsourcing environments. We  provided insights in designing strategies to ameliorate the side impacts of human spammers. Our study can also be applied to analyze and model sophisticated human behavior under different payment mechanisms in many applications. Our future work includes the study of task allocation in crowdsourcing considering the behavioral differences of the crowd. As the security issues of distributed inference systems are becoming increasingly important, we also plan to investigate the robustness of the proposed crowdsourcing algorithm against adversarial attacks.

\appendices	

\section{}\label{prospect}
Let the perception of the variable reward before answering the question be denoted as $Z$. According to the payment mechanism (\ref{payment}), $Z=(\mu_{\max}-\mu_{\min})T^G$ before answering any question. In case that the worker has already provided definitive answers to $i=1,\dots, L$ questions, the expected variable reward is 

\vspace{-0.1 in}
\begin{small}
\begin{equation}
    Z= (\mu_{\max}-\mu_{min})T^G\prod_{i=1}^{L}\frac{t_i}{T}
\end{equation}
\end{small}where $t_i$ is the worker's confidence on the $i^{\mbox{th}}$ question. Hence, we have $Z>0$.

If the worker decides to answer the question, there is a probability of $t$ that the answer is correct and $Z$ is multiplied by $\frac{1}{T}$, and a probability of $1-t$ that the answer is wrong leading to $Z=0$. By using the current expected reward $Z$ as a reference point, and applying the value function to the gains (losses) and the probability weighting function to the probabilities,  the subjective payoff if the worker answers this question is expressed as:

\vspace{-0.05 in}
\begin{small}
\begin{equation}
    SP(t) = w(t)v(Z(\frac{1}{T}-1))+w(1-t)v(-Z)
\end{equation}
\end{small}On the other hand, if the worker skips the question, the expected reward $Z$ stays the same and hence, the subjective payoff is $0$. The worker makes a decision by choosing the action which yields a higher subjective payoff
%\begin{equation}
%SP\quad\mathop{\gtreqless}_{skip}^{answer}\quad 0.
%\end{equation}

\vspace{-0.06 in}
\begin{small}
\begin{align}
{SP}\overset{{answer}}{\underset{{skip}}{\gtrless}}{{0. }}\end{align}
\end{small}which becomes the result of Theorem 1 after simplification.

\section{}\label{paymentA}
If a spammer skips $g$ out of $G$ gold standard questions and answers the remaining $G-g$ with random guesses, the expected monetary reward $E$ for the spammer is expressed as

\vspace{-0.1 in}
\begin{small}
\begin{align}
E&= (\mu_{\max}-\mu_{\min})T^G\prod\limits_{i=1}^G \alpha_{x_i}+\mu_{\min}\nonumber\\
&= (\mu_{\max}-\mu_{\min})T^G(\frac 1 2)^{G-g}(\frac 1 T)^{G-g}\nonumber+\mu_{\min}\\
&= (\mu_{\max}-\mu_{\min})(\frac {1} 2)^G(2T)^g+\mu_{\min},
\end{align}
\end{small}where $X=\{x_1,\ldots,x_G\}$ are the spammer's responses to the gold standard questions. Since $0\le g\le G$, $E$ is maximized as following

\vspace{-0.1 in}
\begin{small}
\begin{align}
&{\text{if}}\ T<\frac 1 {2} \Rightarrow g=0,\ \ 
{\text{if}}\ T>\frac 1 {2} \Rightarrow g=G.
\end{align}
\end{small}

\section{}\label{weightassignment2}
 When there are $M$ spammers in the crowd with $M_0$ skipping and $M_N$ completing all the questions, the expected weight contributed to the correct class is given by
 
 \vspace{-0.1 in}
 \begin{small}
\begin{align}\label{21}
&E_C[\mathbb{W}]\nonumber\\=&\sum\limits_{w=1}^{W-M}E_{p,r}\left[\sum\limits_{n=0}^{N}W_w(n) r(n)\mathbb{P}(n)\right]+\sum\limits_{w=1}^{M_0}W_w(n=0)\nonumber\\
&+\sum\limits_{w=1}^{M_N}\frac 1 {2^N}W_w(n=N)\nonumber\\
=&\sum\limits_{n=0}^N(W-M)W_w(n)\mu^n\binom{N}{n}(1-m)^nm^{N-n}\nonumber\\
&+\sum\limits_{n=0}^NM_0W_w(n)\delta(n)+\sum\limits_{n=0}^N\frac{M_N}{2^N}W_w(n)\delta(n-N)\nonumber\\
=&\sum\limits_{n=0}^N(W\!-\!M)W_w(n)\mu^n\mathbb P(n)\!+\!\sum\limits_{n=0}^N\frac{M_0}{\mathbb P(0)}W_w(n)\mathbb P(n)\delta(n)\nonumber\\
&+\sum\limits_{n=0}^N\frac{M_N}{2^N \mathbb P(N)}W_w(n)\mathbb P(n)\delta(n-N)\nonumber\\
=&\sum\limits_{n=0}^NW_w(n)S(n)\mathbb P(n)
\end{align}
\end{small}where $r(n)$ is the product of any $n$ out of $N$ variables $r_{w,i}$ for $i=1,\dots,N$, which represents the probability that $n$ answers are correct given $n$ definitive answers have been submitted\footnote{Candidate scores are assigned to the correct class only when all the definitive answers are correct.}. $\mathbb P(n)=\binom{N}{n}(1\!-\!m)^nm^{N\!-\!n}$ represents the probability that the $w^{\mbox{th}}$ worker submits a total of $n$ definitive answers. In the last step of (\ref{21}), $S(n)=(W-M)\mu^n+\frac{M_0}{m^N}\delta(n)+\frac{M_N}{2^N(1-m)^N}\delta(n-N)$.

%\vspace{-0.1 in}
%\begin{small}
%\begin{align}
%S(n)=(W-M)\mu^n+\frac{M_0}{m^N}\delta(n)+\frac{M_N}{2^N(1-m)^N}\delta(n-N).
%\end{align}
%\end{small}

Since ${\sum\limits_{n = 0}^N {\mathbb P(n)}  }=1$,  \eqref{21} is upper-bounded according to Cauchy-Schwarz inequality:

\vspace{-0.1 in}
\begin{small}
\begin{align}
&E_C[\mathbb W]=\sum\limits_{n=0}^NW_w(n)S(n)\mathbb P(n)\nonumber\\
\label{24}&\le \sqrt{\sum\limits_{n=0}^N(W_w(n)S(n))^2\mathbb P(n)}\sqrt{\sum\limits_{n=0}^N\mathbb P(n)}=\alpha
\end{align}
\end{small}Also note that equality holds in \eqref{24} only if $W_w(n)S(n)\sqrt{\mathbb P(n)}=\alpha\sqrt{\mathbb P(n)}$, where $\alpha$ is a positive constant such that $W_w(n)S(n)=\alpha$.

Therefore, the optimal weight assignment is obtained 

\vspace{-0.1 in}
\begin{small}
\begin{align}\label{Wwlong}
W_w(n) \!=\! \left[\!\left( {W \!\!-\!\! M} \right){\mu ^n}\!\! +\!\! \frac{{{M_0}}}{{{m^N}}}\delta \left( n \right) \!\!+\! \!\frac{{{M_N}}}{{{2^N}{{\left( {1\! -\! m} \right)}^N}}}\delta \left( {n\! -\! N} \right)\!\right]^{\!-\!1}
\end{align}
\end{small}

Note that the final classification decision $d^*$ corresponds to a unique $N$-bit word, and each bit, $1$ or $0$, represents the decision of a microtask. From the Proposition 1 in \cite{li2016multi}, we know that the classification rule (\ref{first}) is equivalent to the bit-by-bit decision for the $i^{\mbox{th}}$ bit

\vspace{-0.1 in}
\begin{small}
\begin{align}\label{sumtw}
{{\sum\limits_{w=1} ^{W}{T_w}  }}\overset{{{H}}_1}{\underset{{{H}}_0}{\gtrless}}{{0 }}\end{align}
\end{small}for $i=1,\dots,N$, with 
%$T_w={W_w}\left( {{I_1}\left\langle {i,w} \right\rangle  - {I_0}\left\langle {i,w} \right\rangle } \right)$

\vspace{-0.1 in}
\begin{small}
\begin{align}\label{Tw}
T_w={W_w(n)}\left( {{I_1}\left\langle {i,w} \right\rangle  - {I_0}\left\langle {i,w} \right\rangle } \right),
\end{align}
\end{small}where
${I_s}\left\langle {i,w} \right\rangle, s\in\{0,1\},$ is the indicator function that equals $1$ if the answer of the $w^{\mbox{th}}$ worker to the $i^{\mbox{th}}$ question is $s$, and it equals $0$ otherwise.

Hence, if a worker submits no definitive answers and skips all the questions, i.e., $n=0$, his/her decision is not taken into consideration for aggregation at the FC. Thus, we can assign any weight to the worker with $n=0$. Essentially we are neglecting Type I spammers and excluding them for classification.  For simplicity and consistency purpose, we drop the second term on the right hand side of (\ref{Wwlong}) and write the weight assignment as ${W_w(n)} = \left[\left( {W - M} \right){\mu ^n} +% \frac{{{M_0}}}{{{m^N}}}\delta \left( n \right)
\frac{{{M_N}}}{{{2^N}{{\left( {1 - m} \right)}^N}}}\delta \left( {n - N} \right)\right]^{-1}$.
%the corresponding weight assigned is $(W-M+\frac{M_0}{m^N})^{-1}$. However, since this worker skips all the questions, his/her decision for a certain question is not taken into consideration in the fusion center. Thus, we can ignore the weight assignment in such a case and write the scheme as

\section{}\label{pf2}
%First, we write the $i$th bit decision criterion as
%\begin{align}\label{test}
%{{\sum\limits_{w=1} ^{W}{T_w}  }}\overset{{{H}}_1}{\underset{{{H}}_0}{\gtrless}}{{0 }}\end{align}
%with 
%\begin{align}
%T_w={W_w}\left( {{I_1}\left\langle {i,w} \right\rangle  - {I_0}\left\langle {i,w} \right\rangle } \right),
%\end{align}
%where ${I_s}\left\langle {i,w} \right\rangle, s\in\{0,1\},$ serves as the indicator function which equals 1 when the $w$th worker submits $s$ for $i$th bit and otherwise 0.

Following Appendix B, the $i^{\mbox{th}}$ bit of the final aggregated $N$-bit word is determined by (\ref{sumtw}) and $T_w$ is the weighted decision from the $w^{\mbox{th}}$ worker. Let $H_s$ denote the hypothesis that a microtask has true answer to be $s$ for $s=0,1$. If the $w^{\mbox{th}}$ worker is honest, the probability mass function (PMF) of $T_w$ under hypothesis $H_s$, $\Pr \left( {{T_w}|{H_s}} \right)$, is given as

\vspace{-0.1 in}
\begin{small}
\begin{align}\label{tw}
&\Pr \left( {{T_w}=I(-1)^{t+1}W_w(n)|{H_s}} \right) \nonumber \\
&= \left\{ {\begin{array}{*{20}{c}}
	{r_{w,i}^{1 - \left| {s - t} \right|}{\left( {1 - {r _{w,i}}} \right)^{\left| {s - t} \right|}}\varphi_n(w), I = 1}\\
	{p_{w,i},\ \ \ \ \  I = 0}
	\end{array}} \right.,t\in \{0,1\}, 
\end{align}
\end{small}where $I = {I_1}\left\langle {i,w} \right\rangle  + {I_0}\left\langle {i,w} \right\rangle $, $\varphi_n(w) =(1-p_{w,i})\sum\limits_C {\prod\limits_{j = 1, j \ne i}^N {p_{w,j}^{{k_j}}{{\left( {1 - {p_{w,j}}} \right)}^{1 - {k_j}}}} }$ represents the probability that the $w^{\mbox{th}}$ worker gives a definitive answer to the $i^{\mbox{th}}$ question and the total number of definitive answers he/she submitted is $n$. $C$ is defined as the set
		
		\vspace{-0.1 in}
		\begin{small}
\begin{align}
C = \left\{ {\left\{ {{k_1}, \ldots ,{k_{i - 1}},{k_{i + 1}}, \ldots ,{k_N}} \right\}: \sum\limits_{j = 1,j\ne i}^N {{k_j} \!=\! N\! -\! n} } \right\}
\end{align}
\end{small}
with ${k_j} \in \left\{ {0,1} \right\}$ and $n\in\{1,\ldots,N\}$. On the other hand, if the $w^{\mbox{th}}$ worker is a Type II spammer who submits a definitive answer randomly, $T_w$ has the following probability mass function:

\vspace{-0.1 in}
\begin{small}
\begin{align}\label{twtype1}
\Pr (T_w) = \left\{\begin{matrix}
1/2, ~ ~~T_w=W_w(N) \\
-1/2,~~~ T_w=-W_w(N)
\end{matrix}\right.
\end{align}
\end{small}

Under the assumption that hypotheses $H_0$ and $H_1$ are equally likely, the probability of correct classification for the $i^{\mbox{th}}$ bit $P_{c,i}$ is ${P_{c,i}} = \frac{{1 + {P_{d,i}} - {P_{f,i}}}}{2}$, where $P_{d,i}$ is the probability of detection, i.e., deciding the $i^{\mbox{th}}$ bit to be ``1'' when the true bit is ``1'' and $P_{f,i}$ is the probability of false alarm, i.e., deciding the $i^{\mbox{th}}$ bit to be ``1'' when the true bit is ``0''. 

For the honest workers from a total of $W$ workers, let $G_0$ denote the subgroup  that decides ``0'' for $i^{\mbox{th}}$ microtask, $G_1$ the subgroup that decides ``1'' and $G_{\lambda}$ the subgroup that decides $\lambda$. Moreover, out of the $M_N$ Type II spammers we assume that there are $M_N^{\prime}$ spammers deciding ``1'' for the $i^{\mbox{th}}$ bit and $M_N^{\prime \prime}$ deciding ``0''.  We employ the result in (\ref{tw}) and assume that the workers answer the questions independently. Under $H_1$, the probability of the crowd's answer profile for the $i^{\mbox{th}}$ bit is $\{G_0, G_1,G_{\lambda}, M_N^{\prime},M_N^{\prime \prime}\}$ can be expressed as

\vspace{-0.1 in}
\begin{small}
\begin{align}\label{preF}
{ F}_i =& {\left(\frac 1 2\right)}^{M_N^{\prime}} {\left(\frac 1 2\right)}^{M_N^{\prime\prime}}\prod\limits_{w \in {G_{{\lambda}}}} {{p_{w}}}  \prod\limits_{w \in {G_{{0}}}} {\left( {1 - {r _{w,i}}} \right)} {\varphi _{n_w}}\left( w \right) \nonumber\\& \prod\limits_{w \in {G_{{1}}}} {{r _{w,i}}} {\varphi _{n_w}}\left( w \right)
\end{align}
\end{small}where  ${n_w}$ represents the total number definitive questions submitted by the individual.   
Let $q_n, -N\le n\le N$, denote the number of honest workers that submit $|{n}|$ total definitive answers to all the microtasks. Specifically,  $n<0$ indicates the group of honest workers that submit ``0'' for the $i^{\mbox{th}}$ bit while $n>0$ indicates ``1''. For $n=0$, $q_0$ represents the number of honest workers that submit $\lambda$ for the $i^{\mbox{th}}$ bit. Note that the number of honest workers in subgroups $G_0$ , $G_1$ and $G_{\lambda}$ are equal to $\sum_{n=-1}^{-N} q_n$, $\sum_{n=1}^{N} q_n$ and $q_0$, respectively. Denoting

\vspace{-0.1in}
\begin{small}
\begin{align}
{{\mathbb {G}}} = \{ &({q_{ - N}},{q_{ - N + 1}}, \ldots {q_N},M_N^{\prime},M_N^{\prime \prime}):\nonumber\\
&\sum\limits_{n =  - N}^N {{q_n} = W - M_N-M_0 } ,M_N^{\prime}+M_N^{\prime \prime}=M_N \},
\end{align}
\end{small}with natural numbers  $M_N^{\prime}$, $M_N^{\prime \prime}$, and $q_n$ for $\{ n=-N,\dots,0,\dots,N\}$.  From the result in (\ref{preF}),  the answer profile for the $i^{\mbox{th}}$ bit $\mathbb G$ has the following probability under $H_1$

\vspace{-0.1 in}
\begin{small}
	\begin{align}
	&F({{\mathbb G}}) \nonumber=\\& m^{q_0}{\left(\frac{1}{2}\right)}^{M_N}\!\!\prod\limits_{n = 1}^N {{{\left( {1\! -\! \mu } \right)}^{{q_{\! - \!n}}}}{\mu ^{{q_n}}}{{\left(\! {\binom{N\!-\!1}{n\!-\!1}{{\left( {1 \!-\! m} \!\right)}^n}{m^{N - n}}} \right)}^{{q_{ \!-\! n}} \!+\! {q_n}}}} 
	\end{align}
	\end{small}where we substitute the expression of ${\varphi _n}(w)$ using  ${\binom{N-1}{n-1}{{\left( {1 \!-\! m} \right)}^n}{m^{N - n}}}$. Based on the above results, the probability of detection $P_{d,i}$ can be expressed as 
	\vspace{-0.1 in}
	\begin{small}
\begin{align}\label{pdi}
{P_{d,i}} = \sum\limits_S \binom{W,M}{\mathbb{G}} { F_i}(\mathbb{G}) + \frac{1}{2}\sum\limits_{S^\prime} \binom{W,M}{\mathbb{G}}{ F_i}(\mathbb{G}) ,
\end{align}
\end{small}where $\binom{W,M}{\mathbb{G}} = \frac{{(W-M_0)!}}{M_N^{\prime}!M_N^{\prime \prime}!{\prod_{n =  - N}^N {{q_n}!} }}$ represents all possible combinations in the answer profile $\mathbb G$ and 

\vspace{-0.1 in}
\begin{small}
	\begin{align}
	&{S} \!=\! \left\{ {{{\mathbb G}}\!:\!\sum\limits_{n = 1}^N {({q_n} \!-\! {q_{ - n}})W_w(n)}  \!+\!(M_N^{\prime}\!-\!M_N^{\prime \prime})W_w(N) }>0 \right\}\\
	&{S}^\prime \! \!=\!\! \left\{ {{{\mathbb G}}\!:\!\sum\limits_{n = 1}^N {({q_n} \!-\! {q_{ - n}})W_w(n)}  \!+\!(M_N^{\prime}\!-\!M_N^{\prime \prime})W_w(N)  }=0 \right\}
	\end{align}
	\end{small}where $S$ represents the scenario where $\sum_{w=1}^{W}T_w>0$ and ``1'' is decided under $H_1$, and $S^{\prime}$ is the case where $\sum_{w=1}^{W}T_w=0$ and the FC decides ``1'' with probability $1/2$.

Similarly, we can obtain $P_{f,i}$ given  $p_{w,i}$ and $r_{w,i}$ as

\vspace{-0.1 in}
\begin{small}
\begin{align}
{P_{f,i}} = \sum\limits_S \binom{W,M}{\mathbb{G}} { F_i}^\prime({\mathbb G}) + \frac{1}{2}\sum\limits_{S^\prime} \binom{W,M}{\mathbb{G}}{ F_i}^\prime({\mathbb G}).
\end{align}
\end{small}
where

\vspace{-0.1 in}
\begin{small}
	\begin{align}
	&F^{\prime}({\mathbb G})=\nonumber\\&  m^{q_0}{\left(\frac{1}{2}\right)}^{M_N}\!\!\prod\limits_{n = 1}^N {{{\left( {1 \!-\! \mu } \right)}^{{q_n}}}{\mu ^{{q_{ \!-\! n}}}}{{\left(\! {\binom{N\!-\!1}{n\!-\!1}{{\left( {1 \!-\! m} \!\right)}^n}{m^{N - n}}} \right)}^{{q_{ \!-\! n}}\! +\! {q_n}}}} 
	\end{align}	
	\end{small}Then, the expected probability of correct classification for the $i^{\mbox{th}}$ bit $P_{c,i}$ can be obtained as
	
	\vspace{-0.1 in}
	\begin{small}
\begin{align}
{P_{c,i}} &= \frac{1}{2}+ \frac{1}{2}\sum\limits_S {\binom{W,M}{\mathbb{G}}} \left( {F_i\left( \mathbb{G} \right) - F_i^{\prime}\left( \mathbb{G} \right)} \right)\nonumber\\ &+ \frac{1}{4}\sum\limits_{S^\prime} {\binom{W,M}{\mathbb{G}}} \left( {F_i\left( \mathbb{G} \right) - F_i^{\prime} \left( \mathbb{G} \right)} \right)
\end{align}
\end{small}A correct classification result is obtained if and only if all the bits in the $N$-bit word are classified correctly, and recall that the microtasks are completed independently. The probability of correct classification of the final result is given as

\vspace{-0.1 in}
\begin{small}
\begin{align}
P_c=E\left[ {\prod \limits _{i=1}^N p_{c,i}} \right]=\prod \limits _{i=1}^N E\left[ {p_{c,i}} \right]=P_{c,i}^N,
\end{align}
\end{small}where $p_{c,i}$ is the realization of the probability of correct decision for the $i^{\mbox{th}}$ bit.
%where  which can be expressed, due to the independence of the microtasks, as
%\begin{align}
%P_c= \prod \limits _{i=1}^N E\left[ {P_{c,i}} \right].
%\end{align}
%Recall $P_{c,i}$ and we can obtain the expected correct classification probability for $i$th bit $E\left[ {P_{c,i}} \right]$, and
Therefore, the crowdsourcing system has overall  correct classification probability $P_c$ that is given by

\vspace{-0.1 in}
\begin{small}
\begin{align}
P_c=&\Big[\frac{1}{2}+ \frac{1}{2}\sum\limits_S {\binom{W,M}{\mathbb{G}}} \left( {F\left( \mathbb{G} \right) - F^{\prime}\left( \mathbb{G} \right)} \right) \nonumber\\&+ \frac{1}{4}\sum\limits_{S^\prime} {\binom{W,M}{\mathbb{G}}} \left( {F\left( \mathbb{G} \right) - F^{\prime} \left( \mathbb{G} \right)} \right)\Big]^N,
\end{align}
\end{small}

%where
%\begin{align}\label{fq}
%F({\mathbb{Q}}) = {m^{{q_0}}}\prod\limits_{n = 1}^N {{{\left( {1 - \mu } \right)}^{{q_{ - n}}}}{\mu ^{{q_n}}}{{\left( {C_{N - 1}^{n - 1}{{\left( {1 - m} \right)}^n}{m^{N - n}}} \right)}^{{q_{ - n}} + {q_n}}}} 
%\end{align}
%and
%\begin{align}\label{f'q}
%F^{\prime}({\mathbb{Q}}) = {m^{{q_0}}}\prod\limits_{n = 1}^N {{{\left( {1 - \mu } \right)}^{{q_n}}}{\mu ^{{q_{ - n}}}}{{\left( {C_{N - 1}^{n - 1}{{\left( {1 - m} \right)}^n}{m^{N - n}}} \right)}^{{q_{ - n}} + {q_n}}}} .
%\end{align}

\section{}\label{asym}

For an honest worker, the statistic $T_w$ has PMF given in (\ref{tw}) and the expected $T_w$ of an honest worker under $H_1$ is given by

\vspace{-0.1 in}
\begin{small}
\begin{align}
{E_{H1}^{H}}\! =\! \mathbb{E}\left[\sum\limits_{t = 0}^1\! {\sum\limits_{n = 1}^N {{{\left( { \!-\! 1} \right)}^{t + 1}}{W_w(n)}{\left( {{r _{w,i}}} \right)^{t }}{\left( {1\! - \!{r _{w,i}}} \right)^{1\!-\!t}}{\varphi _n(w)}} }\right]
\end{align}
\end{small}Substituting the expression of $W_w(n)$ in (\ref{weight_assign}), we have

\vspace{-0.1 in}
\begin{small}
	\begin{align}
	{E_{H1}^H}&=\left(2\mu-1\right)\sum\limits_{n = 1}^N W_w(n)\binom{N-1}{n-1}{{\left( {1 - m} \right)}^n}{m^{N - n}} \\
	 =& \frac{2\mu-1}{W-M}\sum\limits_{n = 1}^{N-1} \binom{N-1}{n-1}{{\left(\frac {1 - m}{\mu} \right)}^n}{m^{N - n}}\nonumber\\& +\frac{\left(2\mu-1\right)\left(1-m\right)^N}{\left(W-M\right)\mu^N+\frac{M_N}{2^N\left(1-m\right)^N}}\\
	 =& \frac{\left(2\mu-1\right)\left(1-m\right)}{\left(W-M\right)\mu} \left(\frac{1-m}{\mu}+m\right)^{N-1}\nonumber\\&+\frac{\left(2\mu-1\right)\left(1-m\right)^NZ_M}{\left(W-M\right)\mu^N Z_M +M_N},
	\end{align}
	\end{small}where $Z_M={2^N\left(1-m\right)^N}$. The variance of the statistic $T_w$ for an honest worker can be expressed as:
	
	\vspace{-0.1 in}
	\begin{small}
\begin{align}
	V_{H1}^H &\triangleq  \left( {\mathbb{E}\left[ {T_w^2} \right] -{(E_{H_1}^H)}^2} \right)\\
	=& \mathbb{E}\left[\sum\limits_{t = 0}^1 \!{\sum\limits_{n = 1}^N {{{\left(W_w(n)\right)}^2}{\left( {{r _{w,i}}} \right)^{t}}{\left( {1 \!- \!{r _{w,i}}} \right)^{1-t}}{\varphi _n(w)}} } \right]  \!-\!(\!E_{H_1}^H\!)^2\\
	=&\frac{1}{\left(W-M\right)^2}\sum\limits_{n = 1}^{N-1} \binom{N-1}{n-1}{{\left(\frac {1 - m}{\mu^2} \right)}^n}{m^{N - n}}\nonumber\\&+\frac{\left(1-m\right)^N}{\left(\left(W-M\right)\mu^N+\frac{M_N}{2^N\left(1-m\right)^N}\right)^2}-(E_{H_1}^H)^2\\
	=&\frac{1-m}{\left(W-M\right)^2\mu^2}\left(\frac{1-m}{\mu^2}+m\right)^{N-1}\nonumber\\&+\frac{\left(1-m\right)^N Z_M^2}{\left(\left(W-M\right)\mu^N Z_M+M_N \right)^2}-(E_{H_1}^H)^2
	\end{align}
	\end{small}On the other hand, the statistic $T_w$ of Type II spammers has PMF (\ref{twtype1}). In this case, the expected value of $T_w$ under $H_1$ is given by
	
	\vspace{-0.1 in}
	\begin{small}
	\begin{align}
	{E_{H1}^S} = \sum\limits_{t = 0}^1 {\sum\limits_{n = N}^N {{{\left( { - 1} \right)}^{t + 1}}{W_w(n)}{\left( {{\frac 1 2}} \right)^{t }}{\left( { {\frac 1 2}} \right)^{1-t}}} } =0
	\end{align}
	\end{small}For Type II spammers, the variance of $T_w$ under $H_1$ is given by 
	
	\vspace{-0.1 in}
	\begin{small}
	\begin{align}
	V_{H1}^S &\triangleq \left( {\mathbb{E}\left[ {T_w^2} \right] -(E_{H_1}^S)^2}\right)\\
	&= \mathbb{E}\left[\sum\limits_{t = 0}^1 {\sum\limits_{n = N}^N {{{\left(W_w(n)\right)}^2}{\left( {{\frac 1 2}} \right)^{t}}{\left( {\frac 1 2} \right)^{1-t}}} } \right]\\
	&=\frac{1}{\left(\left(W-M\right)\mu^N+\frac{M_N}{2^N\left(1-m\right)^N}\right)^2}\nonumber\\
	&=\frac{Z_M^2}{\left(\left(W-M\right)\mu^N Z_M+M_N\right)^2}
	\end{align}
	\end{small}As the number of workers $W$ increases to infinity, according to the Central Limit Theorem, the statistic $\sum_{w=1}^W T_w$ can be approximated by a Gaussian random variable:
	
	\vspace{-0.1 in}
	\begin{small}
	\begin{align}
	    T_w {\sim}  \left\{\begin{matrix}
H_1: \mathcal{N}(M_1,V_1) \\
H_0: \mathcal{N}(M_0,V_0) 
\end{matrix}\right. ~ \text{as}\ w\rightarrow \infty
	\end{align}
	\end{small}
Under $H_1$, we have obtained the mean and variance of $T_w$ for a single honest worker and a spammer. Note that in the crowd there are $W-M$ honest workers and $M_N$ Type II spammers. Since the workers/spammers complete the tasks independently, we have

\vspace{-0.1 in}
\begin{small}
\begin{align}
M_1=&\left(W-M\right)E_{H1}^H+M_NE_{H1}^S\\
=& \frac{\left(2\mu-1\right)\left(1-m\right)}{\mu} \left(\frac{1-m}{\mu}+m\right)^{N-1}\nonumber\\&+\frac{(W-M)\left(2\mu-1\right)\left(1-m\right)^NZ_M}{\left(W-M\right)\mu^N Z_M +M_N}
\end{align}
\end{small}and

\vspace{-0.1 in}
\begin{small}
	\begin{align}
	V_1=&\left(W-M\right)V_{H1}^H+M_NV_{H1}^S\\
	=&\frac{1-m}{\left(W-M\right)\mu^2}\left(\frac{1-m}{\mu^2}+m\right)^{N-1}\nonumber\\+&\frac{\left((W-M)(1-m)^N+M_N\right) Z_M^2}{\left((W-M)\mu^N Z_M+M_N \right)^2}-\frac{M_1^2}{W-M}	
	\end{align}
	\end{small}From similar procedures, under $H_0$ we can obtain $M_0=-M_1=-\mathcal{M}$ and $V_0=V_1=\mathcal{V}$. 	Since the $i^{\mbox{th}}$ bit is determined via (\ref{sumtw}), it is clear to see that the probability of correct classification of the $i^{\mbox{th}}$ bit is $P_{c,i}= Q(-\frac{\mathcal{M}}{\sqrt{\mathcal{V}}})$. By considering the $N$ bits independently, we obtain the desired result.
	
\bibliographystyle{IEEEtran}
\bibliography{IEEEabrv,ref_Lqw}

\end{document}